\documentclass[a4paper,11pt]{article}
\pdfoutput=1

\usepackage{jcappub}
\usepackage[T1]{fontenc}
\usepackage{amsthm}
\usepackage{enumitem}
\usepackage{multirow}
\usepackage{nicefrac}

\newtheorem{thm}{Theorem}
\newtheorem{prop}[thm]{Proposition}

\theoremstyle{remark}
\newtheorem{remark}[thm]{Remark}
\newtheorem{lemma}[thm]{Lemma}

\DeclareMathOperator{\sech}{sech}

\title{\boldmath Affine ANEC selects the closed FRW branch for geodesically complete cosmology}

\author{Nathan L. Burwig}
\author{and Damien A. Easson}

\affiliation{Department of Physics, Arizona State University,\
Tempe, Arizona 85287, USA}

\emailAdd{nburwig@asu.edu}
\emailAdd{easson@asu.edu}

\abstract{Using the affinely parameterized averaged null energy
condition (ANEC) along radial null geodesics, we prove that,
within the globally nondegenerate
Friedmann--Robertson--Walker (FRW) class with $a(t)>0$
for all cosmic times, non-static flat models and all open models
cannot be both null geodesically complete and ANEC-satisfying.
Bounded oscillatory or cyclic
flat/open models obey the same obstruction; in non-static periodic cases the
negative affine-ANEC contribution accumulates over infinitely many cycles,
giving $I_{\rm ANEC}=-\infty$. The sign obstruction is absent in closed $(k=+1)$ FRW geometry, where positive curvature can offset the negative expansion contribution to the affine ANEC, allowing nonsingular, geodesically complete cosmologies with ordinary NEC-respecting matter. We construct closed-FRW scalar-field models and contrast them with their flat
realizations, which require NEC-violating support. We show that neglecting small positive curvature in a flat
reconstruction of the expansion history can produce an apparent
phantom equation of state, with a percent-level shift at $z=1$. Our main result is a curvature classification of ANEC-compatible eternal FRW
cosmologies within this class: non-static flat and all open models are obstructed, while the
closed branch admits explicit complete realizations, with global de Sitter
as the vacuum representative.
}

\keywords{averaged null energy condition, geodesic completeness, bouncing cosmology, spatial curvature}

\begin{document}
\maketitle
\newpage

\section{Introduction}
Whether the universe had a first moment is fundamentally a question about
past geodesic incompleteness.  The Hawking--Penrose singularity theorems
show that, under broad geometric, causality, and energy-condition
assumptions, cosmological and gravitational-collapse spacetimes contain
incomplete causal geodesics
\cite{Penrose:1964wq,Hawking:1966sx,Hawking:1966jv,Hawking:1967ju,Hawking:1970zqf,Hawking:1973uf}. The Borde--Guth--Vilenkin theorem gives a different incompleteness criterion for spacetimes with positive averaged expansion \cite{Borde:2001nh,Kinney:2023urn,Geshnizjani:2023hyd}.  These results motivate the search for nonsingular or past-complete cosmological
scenarios, including emergent, loitering, bouncing, cyclic, limiting-curvature,
and closed-universe or quantum-cosmology constructions
\cite{Ellis:2002we,Ellis:2003qz,Sahni:1991ks,Linde:1995xm,Linde:2003hc,
Graham:2011nb,Easson:2011zy,Cai:2012va,Battefeld:2014uga,
Brandenberger:2016vhg,Novello:2008ra,Mukhanov1992,Brandenberger1993}.
Several of these models exploit positive curvature, energy-condition
loopholes, or asymptotic/static phases to evade the standard singularity
arguments.

Pointwise energy conditions can be violated in semiclassical physics, but
the ANEC has been established under suitable assumptions in many
quantum-field-theoretic settings, particularly along complete achronal
null geodesics
\cite{Wald:1991xn,Graham:2007va,Hartman:2016lgu,Wall:2009wi,Graham:2005cq,Fewster:2006uf}.
In FRW cosmology, affine parametrization introduces a scale-factor
weighting of the null-energy integral that is central to our curvature
classification.

Here we show that, within the globally nondegenerate
FRW class with $a(t)>0$ for all cosmic times, a non-static flat or open
cosmology cannot be both null geodesically complete and ANEC-satisfying.
This restriction excludes regular extensions of incomplete flat-FRW patches
across an $a\to0$ boundary. Such extensions may respect the NEC but need
not admit a single globally nondegenerate flat-FRW description~\cite{Ferreira:2025extension}.
The proof requires no curvature bounds or additional assumptions about
endpoint behavior, and covers cases in which the ANEC integral does
not converge. For non-static periodic flat/open models, the averaged null
energy diverges to negative infinity.

In the closed \(k=+1\) branch, the curvature term enters the affine ANEC identity with the opposite sign.  This enables positive spatial curvature models to support nonsingular bounces without exotic NEC-violating matter \cite{Molina-Paris:1998xmn,Easson:2026ret}.  The simplest illustration is the scale factor $a(t)=h^{-1}\cosh(ht)$: in closed FRW it identifies global de Sitter and saturates the NEC and ANEC, whereas the same scale factor in flat FRW requires phantom-matter support.  We extend this comparison to a fully analytic scalar-field realization of a quadratic bounce and to a cubic-root model supported by a canonical scalar field.
Related results were presented in Ref.~\cite{Burwig:2025hrr}.
We also quantify the apparent phantom behavior that can arise when small
positive curvature is neglected in a flat reconstruction.

The paper is organized as follows. Section~\ref{sec:geodesic-completeness} presents the null-completeness criterion and unified affine-ANEC identity. Section~\ref{sec:k0-anec} presents the flat/open obstruction, including the open case, while Sec.~\ref{sec:anec-k1} explains how positive curvature avoids the sign obstruction. Section~\ref{sec:examples} develops the motivating closed-versus-flat examples. Sec.~\ref{sec:conclusions} summarizes the resulting curvature classification and discusses curvature-induced phantom mimicry.

\section{Null completeness and the affine ANEC identity}
\label{sec:geodesic-completeness}

Throughout, we consider four-dimensional FRW metrics
$g=-dt^2+a(t)^2\gamma_k$ on $\mathbb R\times\Sigma_k$, with
$a\in C^2(\mathbb R)$ and $a(t)>0$. The spatial metric $\gamma_k$ is
complete, has constant sectional curvature $k$, and is defined on a
three-manifold without boundary. Curvature bounds below refer to scalar
polynomial invariants formed algebraically from the Riemann tensor.

The conditions for geodesic completeness of FRW and generalized FRW
spacetimes were derived in \cite{Lesnefsky:2022fen}, with extensions and
applications in \cite{Easson:2024uxe,Easson:2024fzn}. For the present paper
the essential null statement is that along any nontrivial radial null
geodesic,
\begin{equation}
\lambda(t)-\lambda(t_0)=\int_{t_0}^{t}a(\tau)\,d\tau,
\label{eq:affine-length}
\end{equation}
up to an overall positive affine rescaling. Hence affine completeness to the
future and past is equivalent, respectively, to
\begin{equation}
\int_{t_0}^{+\infty}a(t)\,dt=\infty,
\qquad
\int_{-\infty}^{t_0}a(t)\,dt=\infty.
\label{eq:null-completeness}
\end{equation}
Each example presented in Sec.~\ref{sec:examples} has a scale factor bounded away
from zero for all $t$, ensuring both null and timelike completeness.

We adopt metric signature $(-,+,+,+)$ and take $8\pi G=1$. With Hubble parameter
$H=\dot a/a$, the FRW Einstein equations are
\begin{equation}
H^2+\frac{k}{a^2}=\frac{\rho}{3},
\qquad
\dot H-\frac{k}{a^2}=-\frac{\rho+p}{2}.
\label{eq:FRW-eqs}
\end{equation}
For the normalization Eq.~\eqref{eq:affine-length}, the averaged
null energy becomes
\begin{equation}
I_{\rm ANEC}=\int_{-\infty}^{+\infty}\frac{\rho+p}{a}\,dt \,,
\label{eq:affine-anec}
\end{equation}
where the affine weighting is derived in Appendix~\ref{appaffine} . The integral is
defined by independent endpoint limits,
\begin{equation}
I_{\rm ANEC}=\lim_{T_-\to-\infty,\;T_+\to+\infty} I(T_-,T_+),
\qquad
I(T_-,T_+):=\int_{T_-}^{T_+}\frac{\rho+p}{a}\,dt,
\label{eq:anec-limit}
\end{equation}
whenever the limit exists in $\mathbb R\cup\{-\infty,+\infty\}$.
ANEC requires this limit to exist and be nonnegative, and a symmetric principal
value alone is insufficient for this definition. Smooth-cutoff formulations
of averaged energy conditions are discussed in
Refs.~\cite{Wald:1991xn,Fewster:2010gm}; we use the independent-endpoint
convention throughout.

Using
\[
\frac{\dot H}{a}
=\frac{d}{dt}\left(\frac{H}{a}\right)+\frac{H^2}{a},
\]
integration by parts in Eq.~\eqref{eq:FRW-eqs} gives the unified
finite-interval identity
\begin{equation}
I(T_-,T_+)
=-2\left[\frac{H}{a}\right]_{T_-}^{T_+}
-2\int_{T_-}^{T_+}\frac{H^2}{a}\,dt
+2k\int_{T_-}^{T_+}\frac{dt}{a^3}.
\label{eq:master-anec-identity}
\end{equation}
The boundary and bulk terms are kept together until their separate limits
have been established. Thus the flat geometry has a negative expansion
contribution, open curvature adds a further negative contribution, and closed
curvature supplies a positive one.

\section{ANEC obstruction in flat and open FRW spacetimes}
\label{sec:k0-anec}

The nonnegative expansion integral is
\begin{equation}
P:=\int_{-\infty}^{+\infty}\frac{H^2}{a}\,dt\in[0,+\infty].
\label{eq:bulk-P}
\end{equation}
For a null-complete model, choose $\lambda(0)=0$ in
Eq.~\eqref{eq:affine-length}. Then $\lambda$ ranges over $\mathbb R$.
In affine variables, let
\begin{equation}
A(\lambda)=a(t(\lambda)),\qquad
b=\frac{A'}{A}=\frac{H}{a},\qquad
V=-\frac{k}{A^4},\qquad W=b^2+V,
\label{eq:affine-bulk-variables}
\end{equation}
where prime denotes $d/d\lambda$.
The two nonpositive curvature branches have $V\ge0$. Since
$b'=(\dot H-H^2)/a^2$, the null contraction is
$T_{\mu\nu}\ell^\mu\ell^\nu=-2(b'+W)$. Define
\begin{equation}
D:=\int_{\mathbb R}W\,d\lambda
=\int_{\mathbb R}\left(\frac{H^2}{a}-\frac{k}{a^3}\right)dt
\in[0,+\infty],
\label{eq:nonpositive-bulk}
\end{equation}
so $D=P$ for $k=0$. Write $J(u,v)=I(t(u),t(v))$ for the finite
affine integral. Independent limits in $t$ and $\lambda$ are equivalent
because the change of variables is monotone and covers both full time axes.

\begin{thm}[Flat/open affine-ANEC obstruction]
\label{thm:k0-i}
Let a four-dimensional FRW spacetime satisfying the Einstein equations have
$a\in C^2(\mathbb R)$, $a(t)>0$, and $k=0$ or $-1$, with complete spatial
slices as specified in Sec.~\ref{sec:geodesic-completeness}. Assume null
completeness in both time directions. Then
\begin{equation}
\underline I:=\lim_{L\to\infty}\inf_{x,y\ge L}J(-y,x)\le-2D.
\label{eq:anec-lower-limit}
\end{equation}
If the independent-endpoint ANEC integral exists, it equals
\begin{equation}
I_{\rm ANEC}=-2D\in[-\infty,0].
\label{eq:flat-global}
\end{equation}
It is finite precisely when $D<\infty$ and $b$ tends to zero at both ends.
Either $\underline I$ or an existing $I_{\rm ANEC}$ vanishes exactly for a
static flat model. Every non-static flat model and every open model in this
complete class therefore fails the ANEC condition of
Eq.~\eqref{eq:anec-limit}.
\end{thm}

The proof is given in Appendix~\ref{app:flat-details}. 

The flat/open obstruction requires neither curvature bounds nor endpoint
limits for $H/a$. Static flat FRW is the sole equality case, locally
Minkowski and globally Minkowski for simply connected flat spatial slices.
Null completeness alone is unable to guarantee ANEC convergence: the smooth
example in Appendix~\ref{app:nonconvergent-anec} has bounded algebraic
curvature invariants and finite $P$, but an undefined independent-endpoint
ANEC integral. The lower-limit bound still applies.

\paragraph{Regular affine ends.}
A useful sufficient condition for both completeness and convergence is that
the one-sided limits
\begin{equation}
L_\pm:=\lim_{t\to\pm\infty}\frac{H}{a}
\label{eq:regular-ends}
\end{equation}
exist and are finite. For $u=1/a$, the cosmic derivative
$\dot u=-H/a$ is then globally bounded, say $|\dot u|\le C_0$. Therefore
\begin{equation}
0<u(t)\le u(0)+C_0|t|,\qquad
a(t)\ge\frac{1}{u(0)+C_0|t|}.
\label{eq:regular-end-completeness}
\end{equation}
Both affine lengths diverge. The finite boundary contribution and
nonnegative bulk integral in Eq.~\eqref{eq:master-anec-identity} establish
ANEC existence, so Theorem~\ref{thm:k0-i} applies. If $D$ is finite,
$L_+=L_-=0$ follows from the theorem. Appendix~\ref{app:flat-details}
includes a complete model with divergent $H/a$ endpoints and existing
negative-infinite ANEC, illustrating the broader nature of the obstruction.

Bounded oscillatory models fall under the same barrier. Their
convergence follows from the bounds on $a,H,\dot H$ given in
Appendix~\ref{app:flat-details}, even when $H/a$ continues to oscillate.
Periodic evolution gives a definite infinite value.

\begin{prop}[Periodic flat/open models]
\label{cor:periodic-flat-open}
Let $a(t)$ be a positive, nonconstant, $C^2$, periodic scale factor. The
corresponding flat $(k=0)$ and open $(k=-1)$ FRW spacetimes have bounded
curvature, are null complete in both time directions, and satisfy
$I_{\rm ANEC}=-\infty$.
\end{prop}

Our proof is provided in Appendix~\ref{app:flat-details}. This result
complements cyclic models with a net increase of the scale factor from cycle
to cycle. There the cycle-averaged expansion can remain positive, so a
BGV-type argument already implies past incompleteness
\cite{Kinney:2023urn,Kinney:2026kvw}. Bounded cyclic flat/open models can
evade that mechanism, but their negative affine-ANEC contribution accumulates
over infinitely many cycles.

\paragraph{Open FRW.}
For $k=-1$, the bulk quantity in Theorem~\ref{thm:k0-i} is
\begin{equation}
Q:=\int_{-\infty}^{+\infty}
\left(\frac{H^2}{a}+\frac{1}{a^3}\right)dt\in(0,+\infty].
\label{eq:open-Q}
\end{equation}
Thus $D=Q>0$, so every null-complete open model has a strictly negative
joint lower limit, and any existing ANEC integral equals $-2Q$. Even a
static scale factor has $Q=+\infty$ and $I_{\rm ANEC}=-\infty$.
The regular-end estimate \eqref{eq:regular-end-completeness} and the
bounded-oscillation results of Appendix~\ref{app:flat-details} provide
sufficient conditions for completeness and convergence in this branch.

\section{ANEC and completeness in \texorpdfstring{\(k=+1\)}{k=+1} FRW spacetimes}
\label{sec:anec-k1}

In the closed geometry the curvature contribution in
Eq.~\eqref{eq:master-anec-identity} is positive. Define
\begin{equation}
C:=\int_{-\infty}^{+\infty}a(t)^{-3}\,dt\in(0,+\infty].
\label{eq:bulk-C}
\end{equation}
The relative size and convergence of $C$ and $P$ determine whether this
contribution can compensate the expansion term.

\begin{thm}[Closed-FRW ANEC convergence criteria]
\label{thm:ANEC-criterion}
Let $(M,g)$ be a closed FRW spacetime satisfying the Einstein equations,
with $a\in C^2(\mathbb R)$, $a(t)>0$, and bounded scalar polynomial curvature
invariants. Then $H$ is bounded, $a$ is bounded away from zero, and the
spacetime is null and timelike complete. With $P$ and $C$ defined in
Eqs.~\eqref{eq:bulk-P} and \eqref{eq:bulk-C}:
\begin{enumerate}[label=\textup{(\alph*)}]
\item If $C=+\infty$ and $P<\infty$, then $I_{\rm ANEC}=+\infty$.
\item If $C<\infty$ and $P<\infty$, then $a(t)\to\infty$ at both ends,
$H/a\to0$, and the affine ANEC is finite, with
\begin{equation}
I_{\rm ANEC}=2(C-P).
\label{eq:closed-global}
\end{equation}
\end{enumerate}
\end{thm}

For the proof, see Appendix~\ref{app:anec}. A useful sufficient condition
for case (b) is $\int_{\mathbb R}a^{-1}dt<\infty$, as in exponential-tail
models: the curvature bounds give
\begin{equation}
P\le H_{\max}^2\int_{\mathbb R}\frac{dt}{a},
\qquad
C\le a_{\min}^{-2}\int_{\mathbb R}\frac{dt}{a}.
\label{eq:inverse-tail-bounds}
\end{equation}
Conversely, $C$ diverges on any infinite time end where $a$ is bounded above,
and can diverge even when $a\to\infty$, for example if
$a(t)\sim |t|^p$ with $0<p\le1/3$. When $P$ is finite, case (a) then applies.

\begin{remark}[Closed cyclic and static models]
For a bounded periodic closed model, the ANEC is governed by the period
average of $(\rho+p)/a$. A positive or negative average gives $+\infty$ or
$-\infty$, respectively. If the average vanishes while the integrand is
nonzero, its primitive is periodic and the independent-endpoint limit fails
to exist. Any ordinary-matter realization with $\rho+p\ge0$ satisfies the
affine ANEC directly.

With the usual spherical spatial slices, the static closed geometry is
$\mathbb R\times S^3$, with nonzero curvature.
For $a=a_0$, Eq.~\eqref{eq:FRW-eqs} requires
\[
\rho_{\rm tot}=\frac{3}{a_0^2},\qquad
p_{\rm tot}=-\frac{1}{a_0^2},\qquad
w_{\rm tot}=-\frac13.
\]
This source can be represented by an effective frustrated string network or
by a decomposition that includes vacuum energy. The total density and
pressure are fixed by the geometry, and the ANEC is $+\infty$.
\end{remark}

\section{Energy conditions and motivating examples}\label{sec:examples}

We use the standard perfect-fluid energy conditions as diagnostics for whether a
model requires exotic stress energy or can be supported by ordinary canonical
matter. Throughout,
\begin{equation}
    T_{\mu\nu} = (\rho + p) U_\mu U_\nu + p\, g_{\mu\nu},\label{eq:perfect-fluid}
\end{equation}
where \(\rho\) is the energy density, \(p\) is the isotropic pressure, and
\(U^\mu\) is the fluid four-velocity. For a perfect fluid in four dimensions,
the pointwise conditions reduce to
\begin{align}
\text{WEC:}\qquad &\rho\ge0,\qquad \rho+p\ge0,\\
\text{NEC:}\qquad &\rho+p\ge0,\\
\text{DEC:}\qquad &\rho\ge0,\qquad |p|\le\rho,\\
\text{SEC:}\qquad &\rho+p\ge0,\qquad \rho+3p\ge0.
\end{align}
Equivalently, the DEC requires the energy-current vector
\[
J^\mu=-T^\mu{}_{\nu}t^\nu
\]
to be future-directed causal for every future-directed timelike vector \(t^\mu\).
The logical implications are
\[
\text{DEC}\Rightarrow\text{WEC}\Rightarrow\text{NEC},
\qquad
\text{SEC}\Rightarrow\text{NEC},
\]
with the SEC placing no separate positivity requirement on $\rho$. The NEC is the pointwise condition
most directly tied to the affine ANEC discussed above, while the SEC is the condition generically
violated by accelerated expansion and by a nondegenerate FRW bounce.

Local violations of the pointwise energy conditions are familiar in semiclassical physics, including the Casimir effect, Hawking fluxes, non-minimally coupled fields, and higher-derivative effective theories \cite{Epstein:1965zza,Casimir:1948dh,Barcelo:2002bv,Chatterjee:2012zh,Rubakov:2014jja,Dubovsky:2005xd,Nicolis:2009qm,Kobayashi:2010cm,Deffayet:2010qz,Sawicki:2012pz}. The ANEC is weaker and more global, allowing localized negative null energy but constraining the complete affine average along a null geodesic. It is known to hold in broad classes of quantum field theories in flat spacetime for complete achronal null geodesics, is tied to causality and positivity constraints \cite{Wald:1991xn,Graham:2007va,Hartman:2016lgu}, and helps exclude traversable wormholes, closed timelike curves and related exotic causal structures \cite{Morris:1988tu,Friedman:1993ty,Wall:2009wi,Visser:1996iv}.

Here we impose the ANEC as a classical matter condition along complete
radial null geodesics,
\begin{equation}
\label{anec}
\int_\gamma T_{\mu\nu}k^\mu k^\nu\,d\lambda\ge0,
\end{equation}
where \(\gamma\) is a complete null geodesic, \(k^\mu=dx^\mu/d\lambda\), and
\(\lambda\) is an affine parameter. The rigorous quantum-field-theoretic
theorems motivating ANEC often assume achronality; complete null geodesics in
closed FRW spacetimes need not be globally achronal because of refocusing on
\(S^3\).

Related recent work made use of the smeared null energy condition (SNEC) to constrain
semilocal NEC violation in cosmological settings, including phantom dark energy
and nonsingular bouncing phases~\cite{Moghtaderi:2025cns}.

\subsection{De Sitter bounce \texorpdfstring{\((k=1)\)}{(k=1)}}

Canonical global de Sitter has the closed FRW scale factor
\begin{equation}
a(t)=\frac{1}{h}\cosh(ht),\qquad h>0,
\label{acoshbounce}
\end{equation}
with spatial metric $d\Omega_3^2$ on the unit three-sphere. Its Hubble
quantities are
\begin{equation}
H=h\tanh(ht),\qquad \dot H=h^2\sech^2(ht).
\label{eq:cosh-H}
\end{equation}
Equation~\eqref{eq:FRW-eqs} with $k=+1$ gives $\rho=3h^2$ and
$p=-3h^2$. Thus the NEC and ANEC are saturated, and the bounce is supported
by a positive cosmological constant with $w=-1$. Since
$\ddot a=h\cosh(ht)>0$, the contracting phase slows toward the bounce and
is followed by accelerated expansion. Among the standard pointwise energy
conditions, only the SEC is violated.

\subsection{Transcendental bounce \texorpdfstring{\((k=0)\)}{(k=0)}}

The same scale factor \eqref{acoshbounce} gives a different geometry when
$k=0$.\footnote{This model was discussed in detail as an example of a
\emph{transcendental} bounce in \cite{Easson:2024fzn}.}
Using Eq.~\eqref{eq:cosh-H} in the flat Friedmann equations gives
\begin{equation}
\rho=3h^2\tanh^2(ht),\qquad
\rho+p=-2h^2\sech^2(ht)<0.
\label{rpluspreg}
\end{equation}
Near the bounce, $\rho\to0$ and $p\to-2h^2$, so $w=p/\rho\to-\infty$.
The curvature remains regular:
\begin{equation}
R=6h^2\bigl(1+\tanh^2(ht)\bigr),\qquad
K=12h^4\bigl(\tanh^4(ht)+1\bigr).
\end{equation}
The geometry is geodesically complete, with
\begin{equation}
I_{\rm ANEC}=-2h^3\int_{-\infty}^{+\infty}\sech^3(ht)\,dt
=-\pi h^2.
\label{eq:cosh-flat-anec}
\end{equation}

Only the closed realization is maximally symmetric de Sitter. In the flat
realization the negative null energy in Eq.~\eqref{rpluspreg} requires
matter beyond a real minimally coupled canonical scalar, whose stress tensor
obeys
\begin{equation}
\rho=\tfrac12\dot\phi^2+V(\phi),\qquad
\rho+p=\dot\phi^2\ge0.
\label{scalarmat}
\end{equation}
A \emph{wrong-sign} scalar provides an exact realization, commonly called
\emph{phantom matter}
\cite{Caldwell:1999ew,Carroll:2003st,Caldwell:2003vq,Vikman:2004dc}.\footnote{Higher-derivative matter offers other ways to violate the NEC.}
Its action, profile, field range, and potential are given in
Appendix~\ref{app:phantom-reconstruction}. The NEC, WEC, and DEC fail for
all finite $t$, and the SEC is violated throughout the evolution.
Figure~\ref{fig:cosh_energy_conditions} shows the contrast with the closed
construction. The next two examples realize curvature-supported bounces with
ordinary canonical scalar matter.

\begin{figure}[tbp]
\makebox[\linewidth][c]{\includegraphics[width=\linewidth]{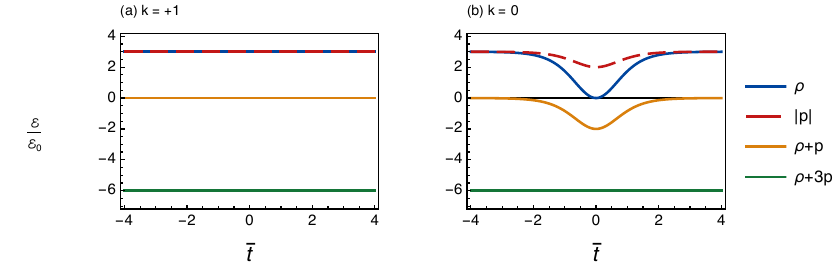}}
\caption{Energy conditions for the same $\cosh$ scale factor
\eqref{acoshbounce} in closed and flat FRW geometries. The vertical quantity
is $\mathcal E/\mathcal E_0$, with
$\mathcal E\in\{\rho,|p|,\rho+p,\rho+3p\}$ and $\mathcal E_0=h^2$;
the horizontal variable is $\bar t=ht$. Panel (a) shows the closed de Sitter
realization, with saturated NEC and ANEC. Panel (b) shows the flat
realization, which requires phantom matter. The parameter choice is $h=1$.}
\label{fig:cosh_energy_conditions}
\end{figure}

\subsection{Quadratic bounce}

Consider the geodesically complete quadratic inflationary-bounce scale factor\footnote{This
model is a subclass of the polynomial bounces discussed in
\cite{Easson:2024fzn}.}
\begin{equation}
a(t)=a_0\left[\left(\frac{t}{\alpha}\right)^2+c\right],
\qquad a_0,\alpha,c>0.
\label{sfquad}
\end{equation}
We set $\alpha=1$ in the following formulas. Then $\ddot a=2a_0>0$,
giving decelerating contraction for $t<0$ and accelerated expansion for
$t>0$. The completeness criterion in Eq.~\eqref{eq:null-completeness} is
satisfied at both ends. The Hubble quantities are
\begin{equation}
H=\frac{2t}{t^2+c},\qquad
\dot H=\frac{2(c-t^2)}{(t^2+c)^2}.
\label{eq:quadratic-H}
\end{equation}
All curvature invariants are bounded. For example, in the flat geometry,
\begin{equation}
R=\frac{12(c+3t^2)}{(t^2+c)^2},\qquad
K=\frac{48(c^2+2ct^2+5t^4)}{(t^2+c)^4}.
\end{equation}
Both vanish as $|t|\to\infty$; closed curvature contributes further finite
terms proportional to powers of $1/a^2$.

For $k=0$, the stress tensor fixed by the geometry\footnote{In comoving
coordinates, $\rho=-G^t{}_t$ and $p=G^i{}_i$ (no sum).} gives
\begin{equation}
\rho+p=\frac{4(t^2-c)}{(t^2+c)^2},\qquad
\rho+3p=-\frac{12}{t^2+c}.
\label{eq:quadratic-flat-conditions}
\end{equation}
The density is nonnegative throughout. The NEC, WEC, and DEC are violated
for $|t|<\sqrt c$ and hold for $|t|\ge\sqrt c$, whereas the SEC is violated
at all times. The affine ANEC is
\begin{equation}
I_{\rm ANEC}
=\frac{4}{a_0}\int_{-\infty}^{+\infty}\frac{t^2-c}{(t^2+c)^3}\,dt
=-\frac{\pi}{a_0c^{3/2}}.
\label{eq:quadratic-flat-anec}
\end{equation}
For the parameters in Fig.~\ref{ecquadplusc}, $a_0=1$ and $c=1/2$, this
is $-2\sqrt2\,\pi$.

\subsubsection{Closed scalar reconstruction}
\label{scalarmat-k1}

For $k=+1$, Eqs.~\eqref{eq:FRW-eqs} and \eqref{scalarmat} instead give
\begin{equation}
\dot\phi^2=\frac{2\left[a_0^{-2}-2(c-t^2)\right]}{(t^2+c)^2},
\qquad
V(t)=\frac{10t^2+2a_0^{-2}+2c}{(t^2+c)^2}.
\label{eq:quadratic-reconstruction}
\end{equation}
The scalar is real for all $t$ precisely when
\begin{equation}
0<c\le\frac{1}{2a_0^2}.
\label{eq:quadratic-admissibility}
\end{equation}
Appendix~\ref{app:quadratic-primitive} gives its exact primitive throughout
this range. At the boundary $c=1/(2a_0^2)$, the solution simplifies to
\begin{equation}
\phi(t)=\ln(1+2a_0^2t^2),\qquad
V(\phi)=4a_0^2\bigl(5e^{-\phi}-2e^{-2\phi}\bigr).
\label{eq:quadratic-boundary}
\end{equation}
Here $\dot\phi=2t/(t^2+c)$ is the smooth sign-changing root. The field
turns at the bounce, and the even function $V(t)$ yields the single-valued
potential in Eq.~\eqref{eq:quadratic-boundary}.

For the admissible range \eqref{eq:quadratic-admissibility}, the canonical
scalar satisfies the NEC and WEC. Its affine-ANEC integrand behaves as
$\dot\phi^2/a\sim4/(a_0t^4)$ at large $|t|$, and integration gives
\begin{equation}
I_{\rm ANEC}=\frac{\pi(3-4a_0^2c)}{4a_0^3c^{5/2}}>0.
\label{eq:quadratic-closed-anec}
\end{equation}
Thus the same complete scale factor requires NEC violation in flat FRW
and admits ordinary scalar support in closed FRW.

\begin{figure}[tbp]
\makebox[\linewidth][c]{\includegraphics[width=\linewidth]{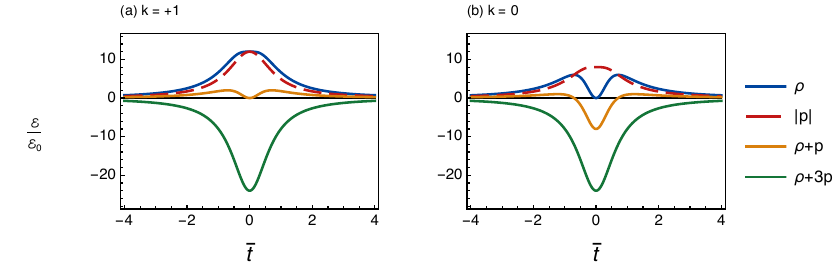}}
\caption{Energy conditions for the quadratic scale factor \eqref{sfquad}.
The vertical quantity is $\mathcal E/\mathcal E_0$, with
$\mathcal E\in\{\rho,|p|,\rho+p,\rho+3p\}$ and
$\mathcal E_0=\alpha^{-2}$; the horizontal variable is
$\bar t=t/\alpha$. Panel (a) shows the closed canonical-scalar realization;
panel (b) shows the flat realization, with NEC violation near the bounce.
The parameters are $a_0=1$, $c=1/2$, and $\alpha=1$.}
\label{ecquadplusc}
\end{figure}

The required SEC violation follows from the curvature-independent
acceleration equation,
\begin{equation}
\frac{\ddot a}{a}=-\frac{\rho+3p}{6}.
\label{eq:acceleration}
\end{equation}
At any nondegenerate FRW bounce, $a>0$ and $\ddot a>0$, so
$\rho+3p<0$. Continuity extends this violation to an open interval around
the bounce. For the quadratic model it persists for all time. A canonical
scalar can satisfy this requirement through potential domination:
$\rho+3p=2\dot\phi^2-2V(\phi)$, while its null energy remains nonnegative.
Positive curvature therefore permits a bounce with NEC- and ANEC-respecting
matter while the SEC is violated.

The quadratic reconstruction provides an explicit starting point for
closed-FRW inflationary models. In particular, the bounce-side potential
in the curvature-bounce inflation construction of Ref.~\cite{Easson:2026ret}
is anchored to this exact scalar reconstruction and then smoothly
interpolated onto a Starobinsky-like slow-roll plateau. It supplies a
closed-form local model of a curvature-supported, NEC-respecting bounce
that can be embedded into a bounce-plus-inflation history.
Table~\ref{tab:examples-summary} summarizes the cosh and quadratic examples.

\begin{table}[t]
\centering
\begin{tabular}{cccc}
\hline\hline
Scale factor & Curvature & Required support & Affine ANEC \\
\hline
$h^{-1}\cosh(ht)$ & $k=0$ & phantom scalar & $-\pi h^2$ \\
$h^{-1}\cosh(ht)$ & $k=+1$ & $\Lambda>0$ & $0$ \\
$a_0(t^2+c)$ & $k=0$ & NEC-violating effective fluid & $<0$ \\
$a_0(t^2+c)$ & $k=+1$ & canonical scalar & $>0$ \\
\hline\hline
\end{tabular}
\caption{Representative complete FRW examples. The spatial curvature fixes
the stress tensor needed to realize a given scale factor. Canonical support
in the closed quadratic row assumes $a_0>0$ and
$0<c\le1/(2a_0^2)$, as in Eq.~\eqref{eq:quadratic-admissibility}.}
\label{tab:examples-summary}
\end{table}

\subsection{Cubic-root bounce}

We next consider a scale factor originally obtained as the isotropic
nonsingular solution of a Born--Infeld--type mimetic construction with a
limiting-curvature constraint
\cite{Markov1982,markov1987,STAROBINSKY198099,Frolov1989,Frolov1990,Mukhanov1992,Brandenberger1993,Chamseddine:2016uef,Chamseddine2017ktu,Frolov:2021afd}:
\begin{equation}
a(t)=a_0(bt^2+c)^{1/3},\qquad a_0,b,c>0.
\label{mukhanovsf}
\end{equation}
Up to the overall scale normalization, the solution in
Ref.~\cite{Chamseddine:2016uef} is recovered with
$b=3\epsilon_m/4$ and $c=1$. In that theory the bounce arises from a
limiting-curvature modification of the Friedmann dynamics. We construct a
closed-FRW realization supported by an ordinary minimally coupled scalar,
where positive spatial curvature supplies the bounce mechanism. These
theories share a background scale factor while differing in perturbations,
ultraviolet completion, and stress-tensor interpretation.

\subsubsection{Scalar-field reconstruction}

Writing $X(t)=bt^2+c$, the Hubble quantities are
\begin{equation}
H=\frac{2bt}{3X},\qquad
\dot H=\frac{2b(c-bt^2)}{3X^2}.
\label{eq:H-quantities}
\end{equation}
The closed Friedmann equations give
\begin{equation}
\dot\phi^2(t)=\frac{2}{3a_0^2}
\left[3X^{-2/3}+2a_0^2b(bt^2-c)X^{-2}\right].
\label{eq:cubic-phi-dot}
\end{equation}
The necessary and sufficient condition for this quantity to be nonnegative
at every time is
\begin{equation}
3c^{1/3}\ge2a_0^2b.
\label{eq:cubic-admissibility}
\end{equation}
Appendix~\ref{app:cubic-details} proves this condition.

In the strict range, the positive root defines the smooth monotonic field
\begin{equation}
\phi(t)=\phi_0+\int_0^t\sqrt{\dot\phi^{\,2}(\tau)}\,d\tau,
\label{eq:cubic-phi-quadrature}
\end{equation}
where the kinetic function under the square root is given explicitly by
Eq.~\eqref{eq:cubic-phi-dot}. At equality in
Eq.~\eqref{eq:cubic-admissibility}, a smooth scalar instead follows the
sign-changing root through the bounce. The potential is
\begin{equation}
V(t)=\frac{2b}{3X}+\frac{2}{a_0^2X^{2/3}}.
\label{eq:cubic-V-of-t}
\end{equation}
Together, the scalar quadrature in the strict case, or its smooth
sign-changing branch at equality, and Eq.~\eqref{eq:cubic-V-of-t} give an
explicit parametric potential $(\phi(t),V(t))$ for all $t\in\mathbb R$.
Appendix~\ref{app:cubic-details} establishes that $V(\phi)$ is smooth and
single-valued through the turning point of the equality branch and verifies
the scalar equation of motion.

\subsubsection{Energy conditions}

For the admissible parameters \eqref{eq:cubic-admissibility}, the density
is positive and the NEC and WEC hold throughout. The NEC saturates at the
bounce in the equality case. The SEC is violated near the bounce, where
$\ddot a>0$, as required by Eq.~\eqref{eq:acceleration}.
The affine-ANEC integrand is pointwise nonnegative and integrable: its
leading curvature contribution falls as $|t|^{-2}$ and the expansion
contribution falls faster. The ANEC is therefore finite and strictly
positive.

This example has $a\sim |t|^{2/3}$, so both $C$ and $P$ converge even though
$\int a^{-1}dt$ diverges. It illustrates case (b) of
Theorem~\ref{thm:ANEC-criterion} beyond the sufficient inverse-scale
condition \eqref{eq:inverse-tail-bounds}. In the corresponding flat model,
NEC violation is required near the bounce. The vanishing $H/a$ endpoints
ensure ANEC convergence, and Theorem~\ref{thm:k0-i} gives its negative value. Figure~\ref{fig:cubic_energy_conditions} compares the
two realizations.

\begin{figure}[tbp]
\makebox[\linewidth][c]{\includegraphics[width=\linewidth]{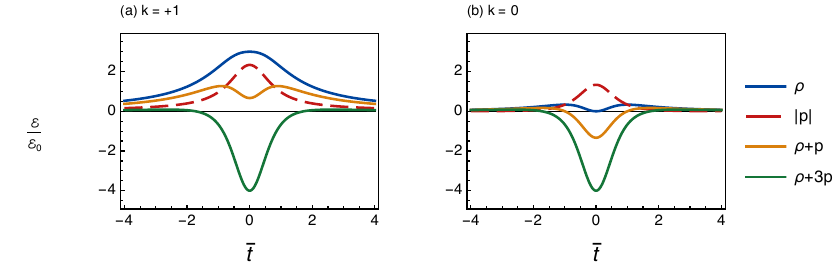}}
\caption{Energy conditions for the cubic-root scale factor \eqref{mukhanovsf}.
The vertical quantity is $\mathcal E/\mathcal E_0$, with
$\mathcal E\in\{\rho,|p|,\rho+p,\rho+3p\}$ and $\mathcal E_0=b$;
the horizontal variable is $\bar t=\sqrt b\,t$. Panel (a) shows the closed
canonical-scalar realization, with positive null energy and SEC violation
near the bounce. Panel (b) shows the flat realization, which requires
NEC-violating effective matter around the bounce. The parameters
$a_0=b=c=1$ lie in the strict admissible range.}
\label{fig:cubic_energy_conditions}
\end{figure}

 \section{Conclusions}
\label{sec:conclusions}

Within the globally nondegenerate FRW class with
$a\in C^2(\mathbb R)$ and $a(t)>0$, null completeness forces a strictly
negative joint lower limit of the affine averages in every non-static flat
model and every open model. Whenever the ANEC integral exists, it is
strictly negative, possibly $-\infty$; the static flat model is the unique
equality case. Non-static periodic flat/open models have
$I_{\rm ANEC}=-\infty$. The smooth example in
Appendix~\ref{app:nonconvergent-anec} shows that completeness and bounded
algebraic curvature invariants do not guarantee ANEC convergence, even
when the bulk integral $P$ is finite.

The closed geometry is different because the positive curvature term in
Eq.~\eqref{eq:master-anec-identity} can compensate the negative expansion
contribution. The cosh, quadratic, and cubic-root examples provide
nonsingular, geodesically complete realizations supported by a cosmological
constant or ordinary canonical scalar matter. They satisfy the NEC and
ANEC, while the SEC is violated at a nondegenerate bounce. Global de
Sitter in closed slicing is the vacuum representative.

The global FRW assumptions distinguish our result from regular extensions
of incomplete flat-FLRW patches across an $a\to0$ boundary, such as those
in Ref.~\cite{Ferreira:2025extension}. At $a=0$ the standard FRW chart is
degenerate; the extended geometry can require different spatial slices and
lies beyond a single globally flat FRW description. Our classification
concerns a single global FRW geometry with a positive scale factor.

Connecting a closed bounce to the standard cosmological history requires
sufficient inflation and a subsequent thermal era. The quadratic scalar
construction provides an analytic starting point for such an evolution.
Earlier studies have found that some closed-scalar bounce models can
require extreme fine-tuning of contracting initial
conditions~\cite{Falciano:2008bounce}. Whether comparable restrictions
apply to the present constructions and their connection to the observed
universe depends on the scalar potential and preceding evolution and
merits further investigation.

A closed primordial universe can evolve to $|\Omega_{k,0}|\ll1$, consistent
with CMB and BAO constraints~\cite{Planck:2018vyg}. Small positive curvature
also affects the inferred dark-energy equation of state: reconstructing
the expansion history of a slightly closed universe assuming flatness
absorbs the negative curvature contribution to $H^2$ into the effective
dark-energy density, potentially producing an apparent phantom equation
of state with $w<-1$
\cite{Clarkson:2007bc,DES:2022ccp}. Appendix~\ref{sec:phantom-mimic}
derives the finite-redshift and local diagnostics at fixed $H_0$ and
$\Omega_{m,0}$. In the illustrative small-curvature example, both give
percent-level shifts at $z=1$ that increase with redshift. Joint
cosmological fits also depend on transverse distances and parameter
degeneracies; DESI DR2 finds that allowing curvature broadens the
dark-energy constraints while retaining the preference for
evolution~\cite{DESI:2025zgx}.

In the globally nondegenerate FRW class, only positive spatial curvature
allows non-static, null-complete cosmologies to satisfy the ANEC.

\acknowledgments
\enlargethispage{\baselineskip}
We thank M.~Dabrowski, P.~Davies, A.~Ferreira, N.~Pinto-Neto, J.~Lesnefsky, V.~Xavier for valuable discussions. DAE is supported in part by the U.S. Department of Energy, Office of High Energy Physics, under Award Number DE-SC0019470.

\appendix

\section{Affine normalization in FRW spacetime}
\label{appaffine}

The FRW line element in comoving coordinates is
\begin{equation}
ds^2=-dt^2+a(t)^2\left(\frac{dr^2}{1-kr^2}+r^2d\Omega^2\right),
\qquad k=0,\pm1.
\end{equation}
For a radial null geodesic, $dr/dt=\pm\sqrt{1-kr^2}/a$. Let
$v=dt/d\lambda>0$, with $\lambda$ an affine parameter. The temporal
geodesic equation, using $\Gamma^t{}_{ij}=Hg_{ij}$ and the null condition,
becomes
\[
\frac{d^2t}{d\lambda^2}+H\left(\frac{dt}{d\lambda}\right)^2=0,
\qquad
\frac{dv}{dt}+Hv=0.
\]
It follows that
\begin{equation}
\frac{dt}{d\lambda}=\frac{E}{a},\qquad
\frac{dr}{d\lambda}=\pm\frac{E\sqrt{1-kr^2}}{a^2},\qquad
d\lambda=\frac{a}{E}\,dt,
\label{eq:affine-normalization}
\end{equation}
where $E>0$ is constant. These local coordinate relations extend along the
null geodesic through changes of spatial chart. Setting $E=1$ gives
Eq.~\eqref{eq:affine-length}.

For the comoving perfect fluid in Eq.~\eqref{eq:perfect-fluid},
$U^\mu=(1,0,0,0)$ and $U_\mu k^\mu=-E/a$. Since $k^\mu$ is null,
\begin{equation}
T_{\mu\nu}k^\mu k^\nu
=(\rho+p)(U_\mu k^\mu)^2
=\frac{E^2}{a^2}(\rho+p).
\end{equation}
Multiplication by $d\lambda$ in Eq.~\eqref{eq:affine-normalization} yields
$E$ times Eq.~\eqref{eq:affine-anec}. Thus the positive affine normalization
changes the magnitude of the integral while preserving its sign.

\section{Proofs for Section~\ref{sec:anec-k1}}
\label{app:anec}

\begin{lemma}[Integrability with bounded logarithmic derivative]
\label{lem:closed-tail}
Let $a:\mathbb R\to(0,\infty)$ be $C^1$, with $|\dot a/a|\le H_{\max}$
on a time end. If $\int a^{-3}dt<\infty$ there, then $a(t)\to\infty$
along that end.
\end{lemma}

\begin{proof}
Consider the future end. If $H_{\max}=0$, $a$ is constant and the integral
diverges. Otherwise, suppose a sequence $t_n\to\infty$ satisfies
$a(t_n)\le M$ for some finite $M$. The bound on $d\log a/dt$ gives
$a(t)\le2M$ for $|t-t_n|\le\delta$, with
$\delta=(\log2)/H_{\max}$. A subsequence gives disjoint such intervals, so
\[
\int^\infty a^{-3}dt
\ge\sum_n\frac{2\delta}{(2M)^3}=+\infty,
\]
a contradiction. Hence $a\to\infty$. The past-end argument is identical.
\end{proof}

\begin{proof}[Proof of Theorem~\ref{thm:ANEC-criterion}]
In closed FRW the Kretschmann scalar is
\begin{equation}
K=12\left[\left(H^2+\frac{1}{a^2}\right)^2
+\left(\dot H+H^2\right)^2\right].
\label{eq:closed-kretschmann}
\end{equation}
Bounded $K$ implies $|H|\le H_{\max}$ and $a\ge a_{\min}>0$. It
bounds $\dot H$ as well. Null completeness follows from
Eq.~\eqref{eq:null-completeness}; a timelike geodesic with conserved
comoving momentum magnitude $\ell$ has
$d\tau/dt=(1+\ell^2/a^2)^{-1/2}$, bounded below by a positive constant.
Its proper time therefore extends to infinity in both directions.

For arbitrary finite truncations, the boundary contribution satisfies
\[
\left|2\left[\frac{H}{a}\right]_{T_-}^{T_+}\right|
\le\frac{4H_{\max}}{a_{\min}}.
\]
In case (a), the positive curvature integral in
Eq.~\eqref{eq:master-anec-identity} diverges while the expansion integral
converges and the boundary term remains bounded. Hence $I_{\rm ANEC}=+\infty$.

In case (b), Lemma~\ref{lem:closed-tail} gives $a(t)\to\infty$ at both
ends. Bounded $H$ then implies $H/a\to0$. Both bulk integrals converge, and
Eq.~\eqref{eq:master-anec-identity} gives Eq.~\eqref{eq:closed-global}.
\end{proof}

\section{Flat/open proofs and endpoint behavior}
\label{app:flat-details}

The proof uses a half-line Riccati estimate we now prove, following the comparison
argument of Galloway~\cite{Galloway:1981minimal} and its formulation in
Fewster and Galloway~\cite{Fewster:2010gm}.

\begin{lemma}[Half-line estimate]
\label{lem:affine-half-line}
Let $f\in C^1([0,\infty))$ and let $v\ge0$ be continuous. Set
$R(x)=\int_0^x(f^2+v)\,ds$ and $F(x)=f(x)+R(x)$.
Then $\limsup_{x\to\infty}F(x)\ge R_\infty$, where
$R_\infty=\lim_{x\to\infty}R(x)\in[0,\infty]$.
\end{lemma}

\begin{proof}
If $R_\infty<\infty$, square integrability gives a sequence
$x_n\to\infty$ with $f(x_n)\to0$. Hence $F(x_n)\to R_\infty$.
If $R_\infty=\infty$ but $\limsup F<\infty$, there is a constant $C_0$
such that $F\le C_0$ eventually. Choose $x_0$ in this tail with
$R(x_0)>C_0+1$. Monotonicity of $R$ gives, for every $x\ge x_0$,
\[
f(x)\le C_0-R(x)<0,\qquad R'(x)\ge(R(x)-C_0)^2.
\]
Consequently
\[
0<\frac{1}{R(x)-C_0}
\le\frac{1}{R(x_0)-C_0}-(x-x_0),
\]
which is impossible at sufficiently large finite $x$. Thus
$\limsup F=\infty$ in the infinite-bulk case.
\end{proof}

\begin{proof}[Proof of Theorem~\ref{thm:k0-i}]
For $x,y\ge0$, introduce
\begin{align}
R_+(x)&=\int_0^xW\,d\lambda,& F_+(x)&=b(x)+R_+(x),\nonumber\\
R_-(y)&=\int_{-y}^0W\,d\lambda,& F_-(y)&=-b(-y)+R_-(y).
\label{eq:affine-two-ends}
\end{align}
The finite-interval identity \eqref{eq:master-anec-identity} becomes
\begin{equation}
J(-y,x)=-2\bigl(F_+(x)+F_-(y)\bigr).
\label{eq:affine-two-end-identity}
\end{equation}
Apply Lemma~\ref{lem:affine-half-line} to $b,V$ on the future end and to
$\widetilde b(y)=-b(-y)$, $\widetilde V(y)=V(-y)$ on the past end.

If $D<\infty$, choose independent sequences $x_n,y_n\to\infty$ on which
$b(x_n),b(-y_n)\to0$. The sum of the two bulk integrals tends to $D$,
so $J(-y_n,x_n)\to-2D$. This proves the lower-limit bound and fixes the
value of any existing full limit.

If $D=\infty$, at least one half-line bulk integral diverges. Its $F$ has
infinite limsup by the lemma. The other $F$ has nonnegative limsup and thus
arbitrarily late values greater than $-1$. Choose independent endpoints
beyond $n$ with the first $F\ge n$ and the other $F>-1$.
Equation~\eqref{eq:affine-two-end-identity} tends to $-\infty$ along this
sequence. Hence $\underline I=-\infty$, and any existing full limit is
$-\infty$.

If that full limit is finite, then $D<\infty$. Subtracting
Eq.~\eqref{eq:affine-two-end-identity} at two future endpoints with a
common late past endpoint shows that $F_+$ is Cauchy at infinity.
The same argument applies to $F_-$. Their finite limits and the finite
limits of $R_\pm$ imply limits of $b$ at both ends; square integrability
forces both to vanish. Conversely, finite $D$ and these zero limits give
finite ANEC directly from the same identity.

Finally, for $k=0$, the continuous nonnegative integrand defining $D=P$
has zero integral exactly when $b\equiv0$, equivalently $a$ is constant.
For $k=-1$, $V=A^{-4}>0$ everywhere, so $D>0$. This proves the equality
case and the strict obstruction.
\end{proof}

\begin{prop}[Bounded oscillatory models]
\label{prop:bounded-oscillatory}
Let $a\in C^2(\mathbb R)$ satisfy
\[
0<a_{\min}\le a(t)\le a_{\max}<\infty,
\qquad |H|,|\dot H|\le M.
\]
The corresponding FRW spacetimes have bounded curvature and are null
complete. Every non-static flat realization has $I_{\rm ANEC}<0$, allowing
$-\infty$, and every open realization has $I_{\rm ANEC}=-\infty$.
\end{prop}

\begin{proof}
The scale-factor bounds give completeness and a bounded boundary term in
Eq.~\eqref{eq:master-anec-identity}. The bounds on $H,\dot H,1/a$ ensure
bounded curvature. For $k=0$, if $P=+\infty$, the negative bulk term forces
$I_{\rm ANEC}=-\infty$. If $P<\infty$, then
\[
\int_{\mathbb R}H^2dt\le a_{\max}P<\infty.
\]
Bounded $\dot H$ makes $H$ uniformly continuous, so $H\to0$ at both ends:
otherwise there would be disjoint intervals of fixed width on which $H^2$
has a positive lower bound. Since $a\ge a_{\min}$, $H/a\to0$ and
Eq.~\eqref{eq:flat-global} gives $I_{\rm ANEC}=-2P<0$ for non-static
evolution. For $k=-1$, the additional negative curvature integral diverges
because $a^{-3}\ge a_{\max}^{-3}$. The bounded boundary term therefore
gives $I_{\rm ANEC}=-\infty$.
\end{proof}

\begin{proof}[Proof of Proposition~\ref{cor:periodic-flat-open}]
A positive $C^2$ periodic scale factor satisfies the bounds in
Proposition~\ref{prop:bounded-oscillatory}. For a nonconstant $a$, the
continuous periodic function $H^2/a$ has a strictly positive integral over
one period. Its integral over an interval tends to $+\infty$ as the interval
length grows, for arbitrary endpoint truncations. Thus $P=+\infty$, and the
preceding proposition gives the result for both $k=0$ and $k=-1$.
\end{proof}

\subsection{Examples of endpoint behavior}
The following smooth examples illustrate the role of the boundary term and
the distinction between the finite-interval identity and its improper limit.
\begin{enumerate}
\item \textbf{Complete, bounded curvature, $I_{\rm ANEC}<0$.}
For $a(t)=(1+t^2)^\alpha$, $\alpha>0$, the limits $H/a$ vanish and
$0<P<\infty$. Equation~\eqref{eq:flat-global} gives negative ANEC.
For $\alpha=1$, $P=\pi/2$ and $I_{\rm ANEC}=-\pi$.
\item \textbf{Incomplete, bounded curvature, $I_{\rm ANEC}=0$.}
Flat-slicing de Sitter, $a(t)=e^{H_0t}$ with $H_0>0$, has $\rho+p=0$.
The boundary and bulk terms diverge at the past end but cancel exactly on
every finite truncation.
\item \textbf{Incomplete, bounded curvature, $I_{\rm ANEC}>0$.}
For $a(t)=e^{-\sqrt{t^2+1}}$, both affine lengths are finite. The boundary
and bulk terms diverge separately, while their sum gives
$I_{\rm ANEC}=+\infty$.
\item \textbf{Complete, bounded curvature, divergent $H/a$ endpoints.}
For
\begin{equation}
a(t)=\frac{1}{\sqrt{1+t^2}\log(e+t^2)},
\label{eq:logarithmic-complete-example}
\end{equation}
$H$ and $\dot H$ are smooth, bounded, and tend to zero at both ends.
As $t\to+\infty$,
\[
a(t)\sim\frac{1}{2t\log t},\qquad
\frac{H}{a}\sim-2\log t,\qquad
\frac{\dot H}{a}\sim\frac{2\log t}{t}.
\]
The affine length diverges like $\tfrac12\log\log t$; evenness gives past
completeness as well. The flat ANEC integrand is eventually negative on
both ends and behaves as $-4\log|t|/|t|$, giving an existing independent-endpoint
limit of $-\infty$. This model lies beyond both the regular-end class and
the class with $a$ bounded away from zero.
\end{enumerate}

\subsection{A complete model with undefined ANEC}
\label{app:nonconvergent-anec}

We construct a flat model with positive smooth scale factor, both infinite
null affine lengths, bounded algebraic curvature invariants, and finite
$D=P$, for which even the symmetric affine average has two different
subsequential limits. The construction is stated in dimensionless units.
Let
\begin{equation}
 \psi(z)=
 \begin{cases}
 \exp\!\left(1-\dfrac{1}{1-z^2}\right),& |z|<1,\\
 0,& |z|\ge1,
 \end{cases}
 \qquad K=4,\quad c=\frac{1}{100},\quad \epsilon_n=\frac{c}{n^2},
 \label{eq:bump}
\end{equation}
and define, for all $\lambda\in\mathbb R$,
\begin{align}
 s(\lambda)&=K\sum_{n=2}^{\infty}
        \psi\!\left(\frac{\lambda-n}{\epsilon_n}\right),&
 b_0(\lambda)&=-\frac{\lambda}{1+\lambda^2},\nonumber\\
 S(\lambda)&=\int_0^\lambda s(u)\,d u,&
 A(\lambda)&=\frac{e^{-S(\lambda)}}{\sqrt{1+\lambda^2}},\nonumber\\
 t(\lambda)&=\int_0^\lambda\frac{\,d u}{A(u)},&
 a(t(\lambda))&=A(\lambda).
 \label{eq:spike-model}
\end{align}
The disjoint bump sum is locally finite, and all these functions are smooth.
Writing $C_j=\int_{-1}^{1}\psi(z)^j\,d z$ for $j=1,2$, one has
\begin{equation}
 0\le S(\lambda)\le S_\infty
 =KcC_1\sum_{n=2}^{\infty}\frac{1}{n^2}<\infty,\qquad
 \frac{e^{-S_\infty}}{\sqrt{1+\lambda^2}}
 \le A(\lambda)\le\frac{1}{\sqrt{1+\lambda^2}}.
 \label{eq:A-bound}
\end{equation}
Here $S=0$ on the negative half-line. The function $t(\lambda)$ is strictly
increasing and tends to both infinite cosmic-time endpoints, so it is a
smooth diffeomorphism of $\mathbb R$. Its inverse obeys $d\lambda=a\,dt$.
This proves the claimed global positivity, smoothness, and null completeness.

\paragraph{Curvature.}
With primes denoting affine derivatives, $b=A'/A=b_0-s$, so
\begin{equation}
 H=Ab,\qquad \dot H=A^2(b'+b^2).
 \label{eq:curvature-check}
\end{equation}
Since $|b|\le K+1/2$ and $A\le1$, $H$ is bounded. On the $n$th bump,
$\lambda\ge n/2$, hence
$A^2/\epsilon_n\le4/c$. Therefore
\[
 |A^2s'|\le\frac{4K}{c}\|\psi'\|_\infty,
 \qquad |b_0'|\le1,
\]
and Eq.~\eqref{eq:curvature-check} bounds $\dot H$ uniformly.
For $k=0$, the orthonormal Riemann components depend only on $H^2$ and
$\dot H+H^2$, so all algebraic scalar polynomial curvature invariants are
bounded. Curvature derivatives grow without bound: at the pulse centers,
$A(n)\sim\alpha/n$ with $\alpha=e^{-S_\infty}$ and
$b''(n)\sim2Kn^4/c^2$, giving
\[
\ddot H=A^3(b''+4bb'+2b^3)
\sim\frac{2K\alpha^3}{c^2}n.
\]
Thus for $\mathcal R=6(\dot H+2H^2)$, the invariant
$\nabla_\mu\mathcal R\nabla^\mu\mathcal R=-\dot{\mathcal R}^{\,2}$ is
unbounded along these centers.

\paragraph{Finite bulk integral.}
Expanding $b^2=(b_0-s)^2$ on the disjoint supports gives
\begin{align}
 P={}&\frac{\pi}{2}
 +2K\sum_{n=2}^\infty\epsilon_n
 \int_{-1}^1\frac{n+\epsilon_n z}{1+(n+\epsilon_n z)^2}\psi(z)\,d z
 +K^2cC_2\sum_{n=2}^\infty\frac{1}{n^2}.
 \label{eq:spike-P}
\end{align}
Using $C_j\le2$ and $2\lambda/(1+\lambda^2)\le1$ for $\lambda>0$ yields
the analytic bound
\begin{equation}
 0<P\le\frac{\pi}{2}
 +2c(K+K^2)\left(\frac{\pi^2}{6}-1\right)<1.829<K.
 \label{eq:P-bound}
\end{equation}

\paragraph{Two different limits.}
At positive integer centers, $b(n)=-n/(1+n^2)-K$, while
$b(-n)=n/(1+n^2)$. In the gaps at $L_n=n+1/2$, there are no bumps at
either endpoint and $b(L_n)-b(-L_n)=-2L_n/(1+L_n^2)$.
Both truncated bulk integrals tend to $P$. Applying Eq.~\eqref{eq:master-anec-identity} in affine variables gives
\begin{equation}
 \lim_{n\to\infty}J(-n,n)=2K-2P>0,\qquad
 \lim_{n\to\infty}J(-L_n,L_n)=-2P<0.
 \label{eq:nonconvergence-subsequences}
\end{equation}
Consequently the improper ANEC integral is undefined, even though curvature
is bounded and the spacetime is null complete. This model is excluded by
the finite-endpoint hypothesis and by a positive uniform lower bound on
$a$. It satisfies the lower-limit obstruction of
Theorem~\ref{thm:k0-i} and is not an ANEC-satisfying counterexample.

\section{Scalar reconstruction details}

\subsection{Phantom support for the flat cosh bounce}
\label{app:phantom-reconstruction}

With signature $(-,+,+,+)$, the wrong-sign scalar action is
\begin{equation}
S[\phi,g]=\int d^4x\,\sqrt{-g}\,
\left(\tfrac12g^{\mu\nu}\partial_\mu\phi\,\partial_\nu\phi-V(\phi)\right).
\end{equation}
Its stress tensor and homogeneous equation of motion are
\begin{align}
T_{\mu\nu}&=-\partial_\mu\phi\,\partial_\nu\phi
+g_{\mu\nu}\left(\tfrac12g^{\alpha\beta}
\partial_\alpha\phi\,\partial_\beta\phi-V\right),\\
\rho&=-\tfrac12\dot\phi^2+V,\qquad
p=-\tfrac12\dot\phi^2-V,\qquad
\ddot\phi+3H\dot\phi-V_{,\phi}=0.
\end{align}
Thus $\rho+p=-\dot\phi^2$.\footnote{The wrong-sign kinetic term is ghostlike: in the minimally coupled
effective theory the Hamiltonian is unbounded below, signaling a vacuum
instability unless the phantom description is an effective parametrization
or is completed by additional degrees of freedom outside the low-energy
model~\cite{Ostrogradsky:1850fid,Cline:2003gs,Woodard:2006nt,Woodard:2015zca}.}
For the flat cosh bounce, Eq.~\eqref{rpluspreg} gives
$\dot\phi=\pm\sqrt2\,h\sech(ht)$. Choosing the positive root yields
\begin{equation}
\phi(t)=2\sqrt2\arctan\!\left[\tanh\!\left(\frac{ht}{2}\right)\right]
+\phi_0,
\qquad
\phi\in\left(\phi_0-\frac{\pi}{\sqrt2},\phi_0+\frac{\pi}{\sqrt2}\right).
\label{eq:phipi}
\end{equation}
The Friedmann density in Eq.~\eqref{rpluspreg} determines
$V(t)=h^2[3-2\sech^2(ht)]$. Using
$\tanh(ht)=\sin[(\phi-\phi_0)/\sqrt2]$ gives
\begin{equation}
V(\phi)=h^2\left[2-\cos\bigl(\sqrt2(\phi-\phi_0)\bigr)\right].
\end{equation}
This potential is smooth and finite on the full field range
\eqref{eq:phipi} and reproduces the flat bounce exactly.

\subsection{Full primitive for the quadratic closed-scalar model}
\label{app:quadratic-primitive}

In the strict range of Eq.~\eqref{eq:quadratic-admissibility}, define
\[
s^2=\frac{1}{2a_0^2}-c>0,\qquad
r=\frac{c}{s^2}=\frac{2a_0^2c}{1-2a_0^2c},\qquad
u=\operatorname{arcsinh}(t/s).
\]
The kinetic term in Eq.~\eqref{eq:quadratic-reconstruction} becomes
$\dot\phi^2=4(t^2+s^2)/(t^2+c)^2$. Taking the positive root and
$\phi(0)=0$ gives the exact primitive
\begin{equation}
\phi(t)=
\begin{cases}
\displaystyle
2u+\frac{2(1-r)}{\sqrt{r(1-r)}}
\arctan\!\left(\sqrt{\frac{1-r}{r}}\tanh u\right),&0<r<1,\\[1.2ex]
2u,&r=1,\\[1.2ex]
\displaystyle
2u+\frac{2(1-r)}{\sqrt{r(r-1)}}
\operatorname{arctanh}\!\left(\sqrt{\frac{r-1}{r}}\tanh u\right),&r>1.
\end{cases}
\label{eq:quadratic-primitive}
\end{equation}
The potential is then given parametrically by substituting $t=s\sinh u$
into Eq.~\eqref{eq:quadratic-reconstruction}. At $s=0$, the smooth
sign-changing branch and its explicit potential are Eq.~\eqref{eq:quadratic-boundary}.

\subsection{Cubic-root admissibility and scalar equation}
\label{app:cubic-details}

To see the admissibility condition, multiply the bracket in
\eqref{eq:cubic-phi-dot} by the positive factor \(X^2\).  The sign of
\(\dot\phi^2\) is controlled by
\[
F(t)=3X^{4/3}+2a_0^2 b\,(bt^2-c),
\qquad X=bt^2+c .
\]
Writing \(y=bt^2\ge0\), this becomes
\[
F(y)=3(y+c)^{4/3}+2a_0^2 b\,(y-c),
\]
and
\[
\frac{dF}{dy}=4(y+c)^{1/3}+2a_0^2 b>0 .
\]
Thus \(F\) is minimized at \(y=0\), i.e. at the bounce, and the necessary and
sufficient condition for \(\dot\phi^2\ge0\) for all \(t\) is
\[
F(0)=c\left(3c^{1/3}-2a_0^2b\right)\ge0,
\]
which is Eq.~\eqref{eq:cubic-admissibility}.

At equality in Eq.~\eqref{eq:cubic-admissibility}, $F(0)=0$ and
$F'(0)>0$. Equation~\eqref{eq:cubic-phi-dot} therefore has the form
$\dot\phi^2=t^2g(t^2)$, where $g$ is smooth and strictly positive for
$t^2\ge0$. Its expansion at the bounce is
\begin{equation}
\dot\phi^2=\frac{28b^2}{9c^2}t^2+O(t^4),\qquad
\phi-\phi_0=\frac{\sqrt7\,b}{3c}t^2+O(t^4)
\label{eq:cubic-boundary-smoothness}
\end{equation}
for the smooth choice $\dot\phi=t\sqrt{g(t^2)}$. In particular,
$d\phi/d(t^2)=\sqrt{g(t^2)}/2>0$. The field is a smooth strictly increasing
function of $t^2$, so the even potential in Eq.~\eqref{eq:cubic-V-of-t}
defines a single-valued smooth $V(\phi)$ through the turning point.

Once \(a(t)\), \(\phi(t)\), and \(V(t)\) satisfy the two Friedmann
equations, stress-tensor conservation follows from the Bianchi identity. For a
homogeneous canonical scalar this conservation law is
\[
\dot\phi\left(\ddot\phi+3H\dot\phi+V_{,\phi}\right)=0 .
\]
Thus, on intervals where \(\dot\phi\neq0\), the scalar equation of motion
follows automatically. At regular turning points it extends by continuity on
the reconstructed branch.

\section{Curvature-induced phantom mimicry: analytic estimate}
\label{sec:phantom-mimic}

DESI DR2 BAO combined with CMB and supernova measurements favors evolving
dark energy in some combinations
\cite{DESI:2025zgx,DESI:2025fii}. In phenomenological parametrizations this
can place the effective equation of state below the phantom divide $w=-1$
over part of the observed redshift range. Genuine phantom matter with
positive density violates the NEC
\cite{Caldwell:1999ew,Carroll:2003st,Caldwell:2003vq,Vikman:2004dc}.
The curvature--dark-energy degeneracy provides a geometric source of an
apparent phantom shift~\cite{Clarkson:2007bc,Planck:2018vyg,DES:2022ccp}.

We use the observational convention
\begin{equation}
\Omega_{k,0}=-\frac{k}{a_0^2H_0^2},
\end{equation}
where $a_0$ here denotes the present scale factor. Positive spatial curvature
therefore has $\Omega_{k,0}<0$. Consider a universe containing pressureless
matter, a cosmological constant, and small curvature:
\begin{equation}
H^2(z)=H_0^2\left[
\Omega_{m,0}(1+z)^3+\Omega_{\Lambda,0}+\Omega_{k,0}(1+z)^2\right].
\label{eq:true-curved-H}
\end{equation}
Reconstruct its expansion history with a flat model at the same $H_0$ and
$\Omega_{m,0}$,
\begin{equation}
H^2_{\rm flat}(z)=H_0^2\left[
\Omega_{m,0}(1+z)^3+
\Omega_{{\rm DE},0}(1+z)^{3(1+w_{\rm eff})}\right].
\label{eq:flat-reconstruction-H}
\end{equation}
Normalization at $z=0$ fixes
\begin{equation}
\Omega_{{\rm DE},0}=\Omega_{\Lambda,0}+\Omega_{k,0}.
\label{eq:effective-DE-normalization}
\end{equation}
Matching the non-matter contributions between $z=0$ and a chosen $z>0$
then gives the finite-redshift constant-$w$ diagnostic
\begin{equation}
w_{\rm eff}(z)=-1+
\frac{1}{3\log(1+z)}
\log\!\left[
\frac{\Omega_{\Lambda,0}+\Omega_{k,0}(1+z)^2}
{\Omega_{\Lambda,0}+\Omega_{k,0}}\right].
\label{eq:weff-log}
\end{equation}
This is the constant equation of state that reproduces the effective
dark-energy density ratio over that interval. For the local equation of
state, the reconstructed density is
$\rho_{{\rm DE},{\rm eff}}\propto\Omega_{\Lambda,0}+\Omega_{k,0}(1+z)^2$,
and conservation gives
\begin{equation}
w_{\rm loc}(z)
=-1+\frac{1+z}{3\rho_{{\rm DE},{\rm eff}}}
\frac{d\rho_{{\rm DE},{\rm eff}}}{dz}
=-1+\frac{2\Omega_{k,0}(1+z)^2}
{3[\Omega_{\Lambda,0}+\Omega_{k,0}(1+z)^2]}.
\label{eq:wloc}
\end{equation}
We use these diagnostics where both the present and reconstructed
dark-energy densities are positive. For closed curvature this requires
\begin{equation}
\Omega_{\Lambda,0}+\Omega_{k,0}>0,
\qquad \Omega_{\Lambda,0}+\Omega_{k,0}(1+z)^2>0.
\label{eq:phantom-density-domain}
\end{equation}
Within this domain both diagnostics lie below $-1$. Their difference reflects
the distinction between an interval-averaged constant $w$ and a local equation
of state.

\subsection{Size of the apparent phantom shift}

For comparison with the phantom depths that can arise in CPL fits,
Eq.~\eqref{eq:weff-log} at $z=1$ can be inverted as
\begin{equation}
\Omega_{k,0}=\Omega_{\Lambda,0}
\frac{2^{3(1+w_{\rm eff})}-1}
{4-2^{3(1+w_{\rm eff})}}.
\label{eq:Ok-required-corrected}
\end{equation}
Table~\ref{tab:curvature-required-corrected} gives the resulting curvature
for $\Omega_{\Lambda,0}=0.70$. A CPL value $w(z)$ is local, whereas
$w_{\rm eff}$ in Eq.~\eqref{eq:weff-log} describes an interval; the table
serves as a comparison of scales.

\begin{table}[t]
\centering
\begin{tabular}{ccc}
\hline\hline
Target $w_{\rm eff}(1)$ & $\Delta w\equiv-(w_{\rm eff}+1)$ & Required $\Omega_{k,0}$ \\
\hline
$-1.10$ & $0.10$ & $-0.041$ \\
$-1.15$ & $0.15$ & $-0.057$ \\
$-1.26$ & $0.26$ & $-0.086$ \\
\hline\hline
\end{tabular}
\caption{Curvature required to reproduce several finite-redshift
constant-$w$ phantom shifts at $z=1$, using
Eq.~\eqref{eq:Ok-required-corrected} with $\Omega_{\Lambda,0}=0.70$. These
magnitudes greatly exceed the curvature intervals obtained from CMB+BAO in
nonflat $\Lambda$CDM~\cite{Planck:2018vyg,DESI:2025zgx}.}
\label{tab:curvature-required-corrected}
\end{table}

For context, Planck TT,TE,EE+lowE+lensing+BAO gives
$\Omega_{k,0}=0.0007\pm0.0019$ at 68\% confidence in the one-parameter
curvature extension of $\Lambda$CDM~\cite{Planck:2018vyg}. DESI DR2+CMB
gives $\Omega_{k,0}=(2.3\pm1.1)\times10^{-3}$ with a 68\% credible interval in
nonflat $\Lambda$CDM~\cite{DESI:2025zgx}. These intervals depend on the
cosmological model and observations included.
Taking $\Omega_{k,0}=-0.004$ as an illustrative small negative curvature
parameter and $\Omega_{\Lambda,0}=0.70$, Eqs.~\eqref{eq:weff-log} and \eqref{eq:wloc}
give at $z=1$
\begin{equation}
\Delta w_{\rm eff}:=-(w_{\rm eff}+1)\simeq8.4\times10^{-3},
\qquad
\Delta w_{\rm loc}:=-(w_{\rm loc}+1)\simeq1.6\times10^{-2}.
\label{eq:phantom-shifts}
\end{equation}
These are percent-level shifts for this choice of curvature and redshift.
The effect grows with $z$ and the diagnostics become singular when the
effective density in Eq.~\eqref{eq:phantom-density-domain} reaches zero.

The estimate concerns expansion-history reconstruction with fixed $H_0$ and
$\Omega_{m,0}$. Curvature enters transverse distances directly, so joint
fits involve additional geometric information and correlated parameter
changes~\cite{Clarkson:2007bc}. In the full DESI DR2 analysis, freeing
curvature increases the uncertainties in $w_0,w_a$ while leaving the
posterior displaced from $\Lambda$CDM~\cite{DESI:2025zgx}. The analytic
calculation explains how a modest apparent phantom shift can arise from
positive curvature in a universe whose matter respects the NEC.

\bibliographystyle{unsrtnat}
\bibliography{references}

@article{Burwig:2025hrr,
    author = "Burwig, Nathan L. and Easson, Damien A.",
    title = "{Open case for a closed universe}",
    eprint = "2510.13971",
    archivePrefix = "arXiv",
    primaryClass = "hep-th",
    doi = "10.1103/mn3v-myzc",
    journal = "Phys. Rev. D",
    volume = "113",
    number = "8",
    pages = "083530",
    year = "2026"
}

@article{Easson:2024uxe,
  author        = {Easson, Damien A. and Lesnefsky, Joseph E.},
  title         = {{Inflationary resolution of the initial singularity}},
  eprint        = {2402.13031},
  archivePrefix = {arXiv},
  primaryClass  = {hep-th},
  doi           = {10.1016/j.physletb.2026.140370},
  journal       = {Phys. Lett. B},
  volume        = {875},
  pages         = {140370},
  year          = {2026}
}

@article{Easson:2024fzn,
  author        = {Easson, Damien A. and Lesnefsky, Joseph E.},
  title         = {{Eternal universes}},
  eprint        = {2404.03016},
  archivePrefix = {arXiv},
  primaryClass  = {hep-th},
  doi           = {10.1103/5mhz-m8bg},
  journal       = {Phys. Rev. D},
  volume        = {112},
  number        = {6},
  pages         = {063545},
  year          = {2025}
}

@article{Lesnefsky:2022fen,
  author        = {Lesnefsky, J. E. and Easson, D. A. and Davies, P. C. W.},
  title         = {{Past-completeness of inflationary spacetimes}},
  eprint        = {2207.00955},
  archivePrefix = {arXiv},
  primaryClass  = {gr-qc},
  doi           = {10.1103/PhysRevD.107.044024},
  journal       = {Phys. Rev. D},
  volume        = {107},
  number        = {4},
  pages         = {044024},
  year          = {2023}
}

@article{Penrose:1964wq,
  author  = {Penrose, Roger},
  title   = {{Gravitational collapse and space-time singularities}},
  doi     = {10.1103/PhysRevLett.14.57},
  journal = {Phys. Rev. Lett.},
  volume  = {14},
  pages   = {57--59},
  year    = {1965}
}

@article{Hawking:1966sx,
  author  = {Hawking, Stephen},
  title   = {{The Occurrence of singularities in cosmology}},
  doi     = {10.1098/rspa.1966.0221},
  journal = {Proc. Roy. Soc. Lond. A},
  volume  = {294},
  pages   = {511--521},
  year    = {1966}
}

@article{Hawking:1966jv,
  author  = {Hawking, Stephen},
  title   = {{The Occurrence of singularities in cosmology. II}},
  doi     = {10.1098/rspa.1966.0255},
  journal = {Proc. Roy. Soc. Lond. A},
  volume  = {295},
  pages   = {490--493},
  year    = {1966}
}

@article{Hawking:1967ju,
  author  = {Hawking, Stephen},
  title   = {{The occurrence of singularities in cosmology. III. Causality and singularities}},
  doi     = {10.1098/rspa.1967.0164},
  journal = {Proc. Roy. Soc. Lond. A},
  volume  = {300},
  pages   = {187--201},
  year    = {1967}
}

@article{Hawking:1970zqf,
  author  = {Hawking, S. W. and Penrose, R.},
  title   = {{The Singularities of gravitational collapse and cosmology}},
  doi     = {10.1098/rspa.1970.0021},
  journal = {Proc. Roy. Soc. Lond. A},
  volume  = {314},
  pages   = {529--548},
  year    = {1970}
}

@book{Hawking:1973uf,
  author    = {Hawking, Stephen W. and Ellis, George F. R.},
  title     = {{The Large Scale Structure of Space-Time}},
  publisher = {Cambridge University Press},
  series    = {Cambridge Monographs on Mathematical Physics},
  year      = {1973},
  doi       = {10.1017/9781009253161}
}

@article{Borde:2001nh,
  author        = {Borde, Arvind and Guth, Alan H. and Vilenkin, Alexander},
  title         = {{Inflationary space-times are incomplete in past directions}},
  eprint        = {gr-qc/0110012},
  archivePrefix = {arXiv},
  reportNumber  = {MIT-CTP-3183},
  doi           = {10.1103/PhysRevLett.90.151301},
  journal       = {Phys. Rev. Lett.},
  volume        = {90},
  pages         = {151301},
  year          = {2003}
}

@article{Kinney:2023urn,
  author        = {Kinney, William H. and Maity, Suvashis and Sriramkumar, L.},
  title         = {{Borde-Guth-Vilenkin theorem in extended de Sitter spaces}},
  eprint        = {2307.10958},
  archivePrefix = {arXiv},
  primaryClass  = {gr-qc},
  doi           = {10.1103/PhysRevD.109.043519},
  journal       = {Phys. Rev. D},
  volume        = {109},
  number        = {4},
  pages         = {043519},
  year          = {2024}
}

@article{Geshnizjani:2023hyd,
  author        = {Geshnizjani, Ghazal and Ling, Eric and Quintin, Jerome},
  title         = {{On the initial singularity and extendibility of flat quasi-de Sitter spacetimes}},
  eprint        = {2305.01676},
  archivePrefix = {arXiv},
  primaryClass  = {gr-qc},
  doi           = {10.1007/JHEP10(2023)182},
  journal       = {JHEP},
  volume        = {10},
  number        = {10},
  pages         = {182},
  year          = {2023}
}

@article{Epstein:1965zza,
  author  = {Epstein, H. and Glaser, V. and Jaffe, A.},
  title   = {{Nonpositivity of energy density in quantized field theories}},
  doi     = {10.1007/BF02749799},
  journal = {Nuovo Cim.},
  volume  = {36},
  pages   = {1016},
  year    = {1965}
}

@article{Casimir:1948dh,
  author  = {Casimir, H. B. G.},
  title   = {{On the attraction between two perfectly conducting plates}},
  journal = {Indag. Math.},
  volume  = {10},
  number  = {4},
  pages   = {261--263},
  year    = {1948}
}

@article{Barcelo:2002bv,
  author        = {Barcelo, Carlos and Visser, Matt},
  title         = {{Twilight for the energy conditions?}},
  eprint        = {gr-qc/0205066},
  archivePrefix = {arXiv},
  doi           = {10.1142/S0218271802002888},
  journal       = {Int. J. Mod. Phys. D},
  volume        = {11},
  pages         = {1553--1560},
  year          = {2002}
}

@article{Rubakov:2014jja,
  author        = {Rubakov, V. A.},
  title         = {{The Null Energy Condition and its violation}},
  eprint        = {1401.4024},
  archivePrefix = {arXiv},
  primaryClass  = {hep-th},
  doi           = {10.3367/UFNe.0184.201402b.0137},
  journal       = {Phys. Usp.},
  volume        = {57},
  pages         = {128--142},
  year          = {2014}
}

@article{Hartman:2016lgu,
  author        = {Hartman, Thomas and Kundu, Sandipan and Tajdini, Amirhossein},
  title         = {{Averaged Null Energy Condition from Causality}},
  eprint        = {1610.05308},
  archivePrefix = {arXiv},
  primaryClass  = {hep-th},
  doi           = {10.1007/JHEP07(2017)066},
  journal       = {JHEP},
  volume        = {07},
  number        = {7},
  pages         = {066},
  year          = {2017}
}

@article{Graham:2007va,
  author        = {Graham, Noah and Olum, Ken D.},
  title         = {{Achronal averaged null energy condition}},
  eprint        = {0705.3193},
  archivePrefix = {arXiv},
  primaryClass  = {gr-qc},
  doi           = {10.1103/PhysRevD.76.064001},
  journal       = {Phys. Rev. D},
  volume        = {76},
  pages         = {064001},
  year          = {2007}
}

@article{Wall:2009wi,
  author        = {Wall, Aron C.},
  title         = {{Proving the Achronal Averaged Null Energy Condition from the Generalized Second Law}},
  eprint        = {0910.5751},
  archivePrefix = {arXiv},
  primaryClass  = {gr-qc},
  doi           = {10.1103/PhysRevD.81.024038},
  journal       = {Phys. Rev. D},
  volume        = {81},
  pages         = {024038},
  year          = {2010}
}

@article{Wald:1991xn,
  author  = {Wald, Robert M. and Yurtsever, U.},
  title   = {{General proof of the averaged null energy condition for a massless scalar field in two-dimensional curved space-time}},
  doi     = {10.1103/PhysRevD.44.403},
  journal = {Phys. Rev. D},
  volume  = {44},
  pages   = {403--416},
  year    = {1991}
}

@article{Graham:2005cq,
  author        = {Graham, Noah and Olum, Ken D.},
  title         = {{Plate with a hole obeys the averaged null energy condition}},
  eprint        = {hep-th/0506136},
  archivePrefix = {arXiv},
  doi           = {10.1103/PhysRevD.72.025013},
  journal       = {Phys. Rev. D},
  volume        = {72},
  pages         = {025013},
  year          = {2005}
}

@article{Fewster:2006uf,
  author        = {Fewster, Christopher J. and Olum, Ken D. and Pfenning, Michael J.},
  title         = {{Averaged null energy condition in spacetimes with boundaries}},
  eprint        = {gr-qc/0609007},
  archivePrefix = {arXiv},
  doi           = {10.1103/PhysRevD.75.025007},
  journal       = {Phys. Rev. D},
  volume        = {75},
  pages         = {025007},
  year          = {2007}
}

@article{Visser:1996iv,
  author        = {Visser, Matt},
  title         = {{Gravitational vacuum polarization. 2: Energy conditions in the Boulware vacuum}},
  eprint        = {gr-qc/9604008},
  archivePrefix = {arXiv},
  doi           = {10.1103/PhysRevD.54.5116},
  journal       = {Phys. Rev. D},
  volume        = {54},
  pages         = {5116--5122},
  year          = {1996}
}

@article{Morris:1988tu,
  author  = {Morris, M. S. and Thorne, K. S. and Yurtsever, U.},
  title   = {{Wormholes, Time Machines, and the Weak Energy Condition}},
  doi     = {10.1103/PhysRevLett.61.1446},
  journal = {Phys. Rev. Lett.},
  volume  = {61},
  pages   = {1446--1449},
  year    = {1988}
}

@article{Friedman:1993ty,
  author        = {Friedman, John L. and Schleich, Kristin and Witt, Donald M.},
  title         = {{Topological censorship}},
  eprint        = {gr-qc/9305017},
  archivePrefix = {arXiv},
  doi           = {10.1103/PhysRevLett.71.1486},
  journal       = {Phys. Rev. Lett.},
  volume        = {71},
  pages         = {1486--1489},
  year          = {1993},
  note          = {[Erratum: Phys. Rev. Lett. 75, 1872 (1995)]}
}

@article{Ellis:2002we,
  author        = {Ellis, George F. R. and Maartens, Roy},
  title         = {{The emergent universe: Inflationary cosmology with no singularity}},
  eprint        = {gr-qc/0211082},
  archivePrefix = {arXiv},
  doi           = {10.1088/0264-9381/21/1/015},
  journal       = {Class. Quant. Grav.},
  volume        = {21},
  pages         = {223--232},
  year          = {2004}
}

@article{Ellis:2003qz,
  author        = {Ellis, George F. R. and Murugan, Jeff and Tsagas, Christos G.},
  title         = {{The Emergent universe: An Explicit construction}},
  eprint        = {gr-qc/0307112},
  archivePrefix = {arXiv},
  doi           = {10.1088/0264-9381/21/1/016},
  journal       = {Class. Quant. Grav.},
  volume        = {21},
  number        = {1},
  pages         = {233--250},
  year          = {2004}
}

@article{Sahni:1991ks,
  author       = {Sahni, Varun and Feldman, Hume and Stebbins, Albert},
  title        = {{Loitering universe}},
  reportNumber = {PRINT-91-0070 (CITA,TORONTO)},
  doi          = {10.1086/170910},
  journal      = {Astrophys. J.},
  volume       = {385},
  pages        = {1--8},
  year         = {1992}
}

@article{Graham:2011nb,
  author        = {Graham, Peter W. and Horn, Bart and Kachru, Shamit and Rajendran, Surjeet and Torroba, Gonzalo},
  title         = {{A Simple Harmonic Universe}},
  eprint        = {1109.0282},
  archivePrefix = {arXiv},
  primaryClass  = {hep-th},
  doi           = {10.1007/JHEP02(2014)029},
  journal       = {JHEP},
  volume        = {02},
  pages         = {029},
  year          = {2014}
}

@article{Molina-Paris:1998xmn,
  author        = {Molina-Paris, Carmen and Visser, Matt},
  title         = {{Minimal conditions for the creation of a Friedman-Robertson-Walker universe from a bounce}},
  eprint        = {gr-qc/9810023},
  archivePrefix = {arXiv},
  doi           = {10.1016/S0370-2693(99)00469-4},
  journal       = {Phys. Lett. B},
  volume        = {455},
  pages         = {90--95},
  year          = {1999}
}

@article{Easson:2011zy,
  author        = {Easson, Damien A. and Sawicki, Ignacy and Vikman, Alexander},
  title         = {{G-Bounce}},
  eprint        = {1109.1047},
  archivePrefix = {arXiv},
  primaryClass  = {hep-th},
  doi           = {10.1088/1475-7516/2011/11/021},
  journal       = {JCAP},
  volume        = {11},
  pages         = {021},
  year          = {2011}
}

@article{Cai:2012va,
  author        = {Cai, Yi-Fu and Easson, Damien A. and Brandenberger, Robert},
  title         = {{Towards a Nonsingular Bouncing Cosmology}},
  eprint        = {1206.2382},
  archivePrefix = {arXiv},
  primaryClass  = {hep-th},
  doi           = {10.1088/1475-7516/2012/08/020},
  journal       = {JCAP},
  volume        = {08},
  pages         = {020},
  year          = {2012}
}

@article{Nicolis:2009qm,
  author        = {Nicolis, Alberto and Rattazzi, Riccardo and Trincherini, Enrico},
  title         = {{Energy's and amplitudes' positivity}},
  eprint        = {0912.4258},
  archivePrefix = {arXiv},
  primaryClass  = {hep-th},
  doi           = {10.1007/JHEP05(2010)095},
  journal       = {JHEP},
  volume        = {05},
  pages         = {095},
  year          = {2010},
  note          = {[Erratum: JHEP 11, 128 (2011)]}
}

@article{Dubovsky:2005xd,
  author        = {Dubovsky, S. and Gregoire, T. and Nicolis, A. and Rattazzi, R.},
  title         = {{Null energy condition and superluminal propagation}},
  eprint        = {hep-th/0512260},
  archivePrefix = {arXiv},
  doi           = {10.1088/1126-6708/2006/03/025},
  journal       = {JHEP},
  volume        = {03},
  pages         = {025},
  year          = {2006}
}

@article{Battefeld:2014uga,
  author        = {Battefeld, D. and Peter, Patrick},
  title         = {{A Critical Review of Classical Bouncing Cosmologies}},
  eprint        = {1406.2790},
  archivePrefix = {arXiv},
  primaryClass  = {astro-ph.CO},
  doi           = {10.1016/j.physrep.2014.12.004},
  journal       = {Phys. Rept.},
  volume        = {571},
  pages         = {1--66},
  year          = {2015}
}

@article{Brandenberger:2016vhg,
  author        = {Brandenberger, Robert and Peter, Patrick},
  title         = {{Bouncing Cosmologies: Progress and Problems}},
  eprint        = {1603.05834},
  archivePrefix = {arXiv},
  primaryClass  = {hep-th},
  doi           = {10.1007/s10701-016-0057-0},
  journal       = {Found. Phys.},
  volume        = {47},
  number        = {6},
  pages         = {797--850},
  year          = {2017}
}

@article{Novello:2008ra,
  author        = {Novello, M. and Bergliaffa, S. E. Perez},
  title         = {{Bouncing Cosmologies}},
  eprint        = {0802.1634},
  archivePrefix = {arXiv},
  primaryClass  = {astro-ph},
  doi           = {10.1016/j.physrep.2008.04.006},
  journal       = {Phys. Rept.},
  volume        = {463},
  pages         = {127--213},
  year          = {2008}
}

@article{Markov1982,
  author  = {Markov, M. A.},
  title   = {{Maximal Curvature Conjecture}},
  journal = {JETP Lett.},
  volume  = {36},
  pages   = {265--269},
  year    = {1982}
}

@article{markov1987,
  author  = {Markov, M. A.},
  title   = {{Pisma Zh. Eksp. Teor. Fiz}},
  journal = {Pisma Zh. Eksp. Teor. Fiz.},
  volume  = {46},
  pages   = {342},
  year    = {1987}
}

@article{STAROBINSKY198099,
  author  = {Starobinsky, A. A.},
  title   = {{A new type of isotropic cosmological models without singularity}},
  doi     = {10.1016/0370-2693(80)90670-X},
  journal = {Phys. Lett. B},
  volume  = {91},
  number  = {1},
  pages   = {99--102},
  year    = {1980}
}

@article{Frolov1989,
  author  = {Frolov, V. P. and Markov, M. A. and Mukhanov, V. F.},
  title   = {{Through A Black Hole Into A New Universe?}},
  journal = {Phys. Lett. B},
  volume  = {216},
  pages   = {272},
  year    = {1989},
  doi     = {10.1016/0370-2693(89)91114-3}
}

@article{Frolov1990,
  author  = {Frolov, V. P. and Markov, M. A. and Mukhanov, V. F.},
  title   = {{Black Holes as Possible Sources of Closed and Semiclosed Worlds}},
  journal = {Phys. Rev. D},
  volume  = {41},
  pages   = {383},
  year    = {1990},
  doi     = {10.1103/PhysRevD.41.383}
}

@article{Mukhanov1992,
  author  = {Mukhanov, V. F. and Brandenberger, R. H.},
  title   = {{A Nonsingular universe}},
  journal = {Phys. Rev. Lett.},
  volume  = {68},
  pages   = {1969},
  year    = {1992},
  doi     = {10.1103/PhysRevLett.68.1969}
}

@article{Brandenberger1993,
  author  = {Brandenberger, R. H. and Mukhanov, V. F. and Sornborger, A.},
  title   = {{A Cosmological theory without singularities}},
  journal = {Phys. Rev. D},
  volume  = {48},
  pages   = {1629},
  year    = {1993},
  doi     = {10.1103/PhysRevD.48.1629},
  note    = {[gr-qc/9303001]}
}

@article{Chamseddine:2016uef,
  author        = {Chamseddine, Ali H. and Mukhanov, Viatcheslav},
  title         = {{Resolving Cosmological Singularities}},
  eprint        = {1612.05860},
  archivePrefix = {arXiv},
  primaryClass  = {gr-qc},
  doi           = {10.1088/1475-7516/2017/03/009},
  journal       = {JCAP},
  volume        = {03},
  number        = {3},
  pages         = {009},
  year          = {2017}
}

@article{Chamseddine2017ktu,
  author  = {Chamseddine, A. H. and Mukhanov, V.},
  title   = {{Nonsingular Black Hole}},
  journal = {Eur. Phys. J. C},
  volume  = {77},
  number  = {3},
  pages   = {183},
  year    = {2017},
  doi     = {10.1140/epjc/s10052-017-4759-z},
  note    = {[arXiv:1612.05861 [gr-qc]]}
}

@article{Frolov:2021afd,
  author        = {Frolov, Valeri P. and Zelnikov, Andrei},
  title         = {{Spherically symmetric black holes in the limiting curvature theory of gravity}},
  eprint        = {2111.12846},
  archivePrefix = {arXiv},
  primaryClass  = {gr-qc},
  doi           = {10.1103/PhysRevD.105.024041},
  journal       = {Phys. Rev. D},
  volume        = {105},
  number        = {2},
  pages         = {024041},
  year          = {2022}
}

@article{Caldwell:1999ew,
  author        = {Caldwell, R. R.},
  title         = {{A Phantom menace?}},
  eprint        = {astro-ph/9908168},
  archivePrefix = {arXiv},
  doi           = {10.1016/S0370-2693(02)02589-3},
  journal       = {Phys. Lett. B},
  volume        = {545},
  pages         = {23--29},
  year          = {2002}
}

@article{Carroll:2003st,
  author        = {Carroll, Sean M. and Hoffman, Mark and Trodden, Mark},
  title         = {{Can the dark energy equation-of-state parameter $w$ be less than $-1$?}},
  eprint        = {astro-ph/0301273},
  archivePrefix = {arXiv},
  doi           = {10.1103/PhysRevD.68.023509},
  journal       = {Phys. Rev. D},
  volume        = {68},
  pages         = {023509},
  year          = {2003}
}

@article{Caldwell:2003vq,
  author        = {Caldwell, Robert R. and Kamionkowski, Marc and Weinberg, Nevin N.},
  title         = {{Phantom energy and cosmic doomsday}},
  eprint        = {astro-ph/0302506},
  archivePrefix = {arXiv},
  doi           = {10.1103/PhysRevLett.91.071301},
  journal       = {Phys. Rev. Lett.},
  volume        = {91},
  pages         = {071301},
  year          = {2003}
}

@article{Vikman:2004dc,
  author        = {Vikman, Alexander},
  title         = {{Can dark energy evolve to the phantom?}},
  eprint        = {astro-ph/0407107},
  archivePrefix = {arXiv},
  doi           = {10.1103/PhysRevD.71.023515},
  journal       = {Phys. Rev. D},
  volume        = {71},
  pages         = {023515},
  year          = {2005}
}

@article{Clarkson:2007bc,
  author        = {Clarkson, Chris and Cortes, Marina and Bassett, Bruce A.},
  title         = {{Dynamical Dark Energy or Simply Cosmic Curvature?}},
  eprint        = {astro-ph/0702670},
  archivePrefix = {arXiv},
  doi           = {10.1088/1475-7516/2007/08/011},
  journal       = {JCAP},
  volume        = {08},
  pages         = {011},
  year          = {2007}
}

@article{DESI:2025zgx,
  author        = {Abdul Karim, M. and others},
  collaboration = {DESI},
  title         = {{DESI DR2 results. II. Measurements of baryon acoustic oscillations and cosmological constraints}},
  eprint        = {2503.14738},
  archivePrefix = {arXiv},
  primaryClass  = {astro-ph.CO},
  doi           = {10.1103/tr6y-kpc6},
  journal       = {Phys. Rev. D},
  volume        = {112},
  number        = {8},
  pages         = {083515},
  year          = {2025}
}

@article{Planck:2018vyg,
  author        = {Aghanim, N. and others},
  collaboration = {Planck},
  title         = {{Planck 2018 results. VI. Cosmological parameters}},
  eprint        = {1807.06209},
  archivePrefix = {arXiv},
  primaryClass  = {astro-ph.CO},
  doi           = {10.1051/0004-6361/201833910},
  journal       = {Astron. Astrophys.},
  volume        = {641},
  pages         = {A6},
  year          = {2020},
  note          = {[Erratum: Astron. Astrophys. 652, C4 (2021)]}
}

@article{DES:2022ccp,
  author        = {Abbott, T. M. C. and others},
  collaboration = {DES},
  title         = {{Dark Energy Survey Year 3 results: Constraints on extensions to $\Lambda$CDM with weak lensing and galaxy clustering}},
  eprint        = {2207.05766},
  archivePrefix = {arXiv},
  primaryClass  = {astro-ph.CO},
  doi           = {10.1103/PhysRevD.107.083504},
  journal       = {Phys. Rev. D},
  volume        = {107},
  number        = {8},
  pages         = {083504},
  year          = {2023}
}

@article{Chatterjee:2012zh,
  author        = {Chatterjee, Saugata and Easson, Damien A. and Parikh, Maulik},
  title         = {{Energy conditions in the Jordan frame}},
  eprint        = {1212.6430},
  archivePrefix = {arXiv},
  primaryClass  = {gr-qc},
  doi           = {10.1088/0264-9381/30/23/235031},
  journal       = {Class. Quant. Grav.},
  volume        = {30},
  pages         = {235031},
  year          = {2013}
}

@article{Kobayashi:2010cm,
  author        = {Kobayashi, Tsutomu and Yamaguchi, Masahide and Yokoyama, Jun'ichi},
  title         = {{G-inflation: Inflation driven by the Galileon field}},
  eprint        = {1008.0603},
  archivePrefix = {arXiv},
  primaryClass  = {hep-th},
  doi           = {10.1103/PhysRevLett.105.231302},
  journal       = {Phys. Rev. Lett.},
  volume        = {105},
  pages         = {231302},
  year          = {2010}
}

@article{Deffayet:2010qz,
  author        = {Deffayet, Cedric and Pujolas, Oriol and Sawicki, Ignacy and Vikman, Alexander},
  title         = {{Imperfect Dark Energy from Kinetic Gravity Braiding}},
  eprint        = {1008.0048},
  archivePrefix = {arXiv},
  primaryClass  = {hep-th},
  doi           = {10.1088/1475-7516/2010/10/026},
  journal       = {JCAP},
  volume        = {10},
  pages         = {026},
  year          = {2010}
}

@article{Sawicki:2012pz,
  author        = {Sawicki, Ignacy and Vikman, Alexander},
  title         = {{Hidden Negative Energies in Strongly Accelerated Universes}},
  eprint        = {1209.2961},
  archivePrefix = {arXiv},
  primaryClass  = {astro-ph.CO},
  doi           = {10.1103/PhysRevD.87.067301},
  journal       = {Phys. Rev. D},
  volume        = {87},
  pages         = {067301},
  year          = {2013}
}

@misc{Kinney:2026kvw,
    author = "Kinney, William H.",
    title = "{Geodesic Completeness in General Cosmological Scenarios}",
    eprint = "2604.20809",
    archivePrefix = "arXiv",
    primaryClass = "gr-qc",
    month = "4",
    year = "2026"
}

@article{DESI:2025fii,
    author = "Lodha, K. and others",
    collaboration = "DESI",
    title = "{Extended dark energy analysis using DESI DR2 BAO measurements}",
    eprint = "2503.14743",
    archivePrefix = "arXiv",
    primaryClass = "astro-ph.CO",
    reportNumber = "FERMILAB-PUB-25-0164-PPD",
    doi = "10.1103/w4c6-1r5j",
    journal = "Phys. Rev. D",
    volume = "112",
    number = "8",
    pages = "083511",
    year = "2025"
}

@article{Linde:1995xm,
    author = "Linde, Andrei D.",
    title = "{Inflation with variable Omega}",
    eprint = "hep-th/9503097",
    archivePrefix = "arXiv",
    reportNumber = "SU-ITP-95-5",
    doi = "10.1016/0370-2693(95)00370-Z",
    journal = "Phys. Lett. B",
    volume = "351",
    pages = "99--104",
    year = "1995"
}

@article{Linde:2003hc,
    author = "Linde, Andrei D.",
    title = "{Can we have inflation with Omega {\ensuremath{>}} 1?}",
    eprint = "astro-ph/0303245",
    archivePrefix = "arXiv",
    doi = "10.1088/1475-7516/2003/05/002",
    journal = "JCAP",
    volume = "05",
    pages = "002",
    year = "2003"
}

@article{Easson:2026ret,
    author = "Easson, Damien A.",
    title = "{Geodesically complete curvature-bounce inflation}",
    eprint = "2604.27103",
    archivePrefix = "arXiv",
    primaryClass = "astro-ph.CO",
    doi = "10.1088/1475-7516/2026/08/040",
    journal = "JCAP",
    volume = "08",
    pages = "040",
    year = "2026"
}

@article{Moghtaderi:2025cns,
    author = "Moghtaderi, Elly and Hull, Brayden R. and Quintin, Jerome and Geshnizjani, Ghazal",
    title = "{How much null-energy-condition breaking can the Universe endure?}",
    eprint = "2503.19955",
    archivePrefix = "arXiv",
    primaryClass = "gr-qc",
    doi = "10.1103/j6hp-p8hs",
    journal = "Phys. Rev. D",
    volume = "111",
    number = "12",
    pages = "123552",
    year = "2025"
}

@article{Cline:2003gs,
    author = "Cline, James M. and Jeon, Sangyong and Moore, Guy D.",
    title = "{The Phantom menaced: Constraints on low-energy effective ghosts}",
    eprint = "hep-ph/0311312",
    archivePrefix = "arXiv",
    reportNumber = "MCGILL-03-25",
    doi = "10.1103/PhysRevD.70.043543",
    journal = "Phys. Rev. D",
    volume = "70",
    pages = "043543",
    year = "2004"
}

@article{Woodard:2006nt,
    author = "Woodard, Richard P.",
    editor = "Papantonopoulos, Lefteris",
    title = "{Avoiding dark energy with 1/r modifications of gravity}",
    eprint = "astro-ph/0601672",
    archivePrefix = "arXiv",
    reportNumber = "UFIFT-QG-06-02",
    doi = "10.1007/978-3-540-71013-4_14",
    journal = "Lect. Notes Phys.",
    volume = "720",
    pages = "403--433",
    year = "2007"
}

@article{Woodard:2015zca,
    author = "Woodard, Richard P.",
    title = "{Ostrogradsky's theorem on Hamiltonian instability}",
    eprint = "1506.02210",
    archivePrefix = "arXiv",
    primaryClass = "hep-th",
    reportNumber = "UFIFT-QG-15-03",
    doi = "10.4249/scholarpedia.32243",
    journal = "Scholarpedia",
    volume = "10",
    number = "8",
    pages = "32243",
    year = "2015"
}

@article{Ostrogradsky:1850fid,
  author  = {Ostrogradsky, M.},
  title   = {M\'emoires sur les \'equations diff\'erentielles, relatives au probl\`eme des isop\'erim\`etres},
  journal = {Mem. Acad. St. Petersbourg},
  volume  = {6},
  number  = {4},
  pages   = {385--517},
  year    = {1850}
}

@article{Ferreira:2025extension,
  author        = {Ferreira, Jr., Alexsandre L. and Pinto-Neto, Nelson and Xavier, Vanessa N.},
  title         = {{Bouncing completion of eternal inflation}},
  eprint        = {2504.01224},
  archivePrefix = {arXiv},
  primaryClass  = {gr-qc},
  doi           = {10.1103/rl3f-v3p3},
  journal       = {Phys. Rev. D},
  volume        = {111},
  number        = {12},
  pages         = {123531},
  year          = {2025}
}

@article{Falciano:2008bounce,
  author        = {Falciano, Felipe T. and Lilley, Marc and Peter, Patrick},
  title         = {{A classical bounce: constraints and consequences}},
  eprint        = {0802.1196},
  archivePrefix = {arXiv},
  primaryClass  = {gr-qc},
  doi           = {10.1103/PhysRevD.77.083513},
  journal       = {Phys. Rev. D},
  volume        = {77},
  pages         = {083513},
  year          = {2008}
}

@article{Galloway:1981minimal,
  author  = {Galloway, Gregory J.},
  title   = {{Some results on the occurrence of compact minimal submanifolds}},
  journal = {Manuscripta Mathematica},
  volume  = {35},
  pages   = {209--219},
  year    = {1981},
  doi     = {10.1007/BF01168457}
}

@article{Fewster:2010gm,
  author        = {Fewster, Christopher J. and Galloway, Gregory J.},
  title         = {{Singularity theorems from weakened energy conditions}},
  journal       = {Class. Quant. Grav.},
  volume        = {28},
  pages         = {125009},
  year          = {2011},
  eprint        = {1012.6038},
  archivePrefix = {arXiv},
  primaryClass  = {gr-qc}
}
\end{document}